\documentclass[11pt]{article}
\usepackage[ruled,lined]{algorithm2e} 
\usepackage{lineno,hyperref}

\newcommand{\ignore}[1]{}
\usepackage{graphicx}
\usepackage{mdframed}
\usepackage{subfigure,comment}
\usepackage[title]{appendix}
\usepackage{epsfig}
\usepackage{color}
\usepackage{amsmath}
\usepackage{epstopdf}
\usepackage{epsfig}

\newcommand{\nn}{{\cal N}}
\newcommand{\ee}{{\cal E}}
\newcommand{\att}{{\cal A}}

\definecolor{herecolor}{rgb}{0.1,0.6,0.1 }

\newcommand{\bartb}[1]{\textcolor{black}{#1}}
\newcommand{\bartwashere}[1]{\bigskip
\textcolor{herecolor}{\sc **** Bart has read till here ****}\bigskip\newline }

 \newcommand{\ab}{\allowbreak}
 
\newtheorem{definition}{Definition}
\newtheorem{example}{Example}

\newtheorem{theorem}{Theorem}

\newtheorem{proposition}{Proposition}

\newtheorem{proof}{Proof}
\newtheorem{property*}{Property}
\newtheorem{lemma*}{Lemma}
\newtheorem{query}{Query}
\newtheorem{remark}{Remark}
\def\squareforqed{\hbox{\rlap{$\sqcap$}$\sqcup$}}
\def\qed{\ifmmode\squareforqed\else{\unskip\nobreak\hfil
\penalty50\hskip1em\null\nobreak\hfil\squareforqed
\parfillskip=0pt\finalhyphendemerits=0\endgraf}\fi}
\usepackage{listings}
\lstdefinestyle{tiny}{
  basicstyle=\sffamily\small, 
  numbers=none
  }

\newcommand{\out}[1]{}

\usepackage{makeidx}  
 \makeindex          

\date{}
\begin{document}

\title{Online Analytical Processsing on Graph Data\footnote{This is a draft version of the work
that will appear in Volume 24(2) of the Intelligent Data Analysis Journal, in early 2020.}}
\author{Leticia G\'omez${}^1$, Bart Kuijpers${}^2$,  
Alejandro Vaisman${}^3$}
\maketitle

\begin{abstract}
Online Analytical Processing (OLAP) comprises tools and algorithms that allow querying multidimensional  databases. It is based 
on the multidimensional model, where  data can be seen as a cube such that each cell contains one or more  measures that can be aggregated along  dimensions. 
In a ``Big Data'' scenario, traditional data warehousing and OLAP operations   are  clearly not sufficient to address    current data analysis requirements, for example,   social network analysis. Furthermore,   OLAP operations and models can expand the possibilities of  
graph analysis beyond  the traditional graph-based computation. Nevertheless, there is not much work on the problem of taking 
OLAP analysis to the graph data model. 
 
This paper proposes a formal multidimensional model for graph analysis, 
that considers  the basic graph data, and also
background information in the form of dimension hierarchies. The graphs in this  model are
node- and edge-labelled directed multi-hypergraphs,  called \emph{graphoids}, which can 
be defined at several different levels of granularity using the dimensions
associated with them.  Operations analogous to  the ones used in typical
OLAP over  cubes  are defined  over graphoids.  The paper presents a formal definition of the graphoid model for OLAP,   proves that the typical OLAP operations on cubes can be expressed over the graphoid model, and shows that the classic data cube model is a particular case of the graphoid  data model. Finally, a case study supports the claim that, 
for many kinds of OLAP-like analysis on graphs,  the graphoid model works better than 
the typical relational OLAP alternative, and for the classic OLAP queries, it remains competitive. 
 \end{abstract}

  
\footnotetext[1]{Instituto Tecnol{\'o}gico de Buenos Aires, Buenos Aires, Argentina; email:  lgomez@itba.edu.ar}
\footnotetext[2]{Hasselt University, Belgium; email: bart.kuijpers@uhasselt.edu}
\footnotetext[3]{Instituto Tecnol{\'o}gico de Buenos Aires, Buenos Aires, Argentina; email:  avaisman@itba.edu.ar (Corresponding author).
}

\section{Introduction} \label{sec:introduction}
Online Analytical Processing(OLAP)~\cite{Kimball1996,VZ14} comprises
tools and algorithms that allow  querying
multidimensional (MD) databases.   
In these databases, data are modelled as {\em data cubes}, 
where each cell contains one or more \textit{measures}
of interest, that quantify \textit{facts}. Measure values can be aggregated
along \textit{dimensions},  organized as   sets of
hierarchies.  Traditional OLAP operations are used to manipulate the data cube, for example:   aggregation and disaggregation of 
measure data along the dimensions; selection of a portion of the cube; or projection of
 the data cube over a subset of its dimensions. The cube is computed after a process called ETL, an acronym for Extract, Transform, and Load,
 which requires a complex and expensive load of work to carry data from the sources to the MD database, 
 typically  a data warehouse (DW). Although  OLAP   has been used for social 
 network analysis~\cite{KraiemFKRT15,RehmanWS13},  in a ``Big Data''  
 scenario, further  requirements appear~\cite{OLAPonBigData}. 
 In the classic paper by Cohen et al.~\cite{CohenDDHW09},  the so-called MAD  skills
 (standing from Magnetic, Agile and Deep) required for data analytics are described.  
 In this scenario, more complex analysis tools are required, that go beyond 
 classic OLAP~\cite{DBLP:conf/sigmod/TangHYDZ17}.    
 Graphs, and, particularly, property graphs~\cite{Hartig14,Robinson13},
 are becoming increasingly popular to model different kinds of networks
 (for instance, social networks, sensor networks, and the kind). Property graphs underlie 
 the most popular graph databases~\cite{Angles2012}.  
 Examples of graph databases and graph processing frameworks  following this model are Neo4j\footnote{\url{http:// www.neo4j.com}}, Janusgraph\footnote{\url{http://janusgraph.org/}} (previously called Titan),  and GraphFrames\footnote{\url{https://graphframes.github.io/}}. In addition to traditional graph  analytics, it is also   
 interesting for the data scientist to 
 have the possibility of performing  OLAP on graphs.  
 
 From the discussion above, it follows that, on the one hand, traditional 
 data warehousing  and OLAP operations  on cubes are 
 clearly not sufficient to address the  current data analysis requirements; on the other hand, 
 OLAP operations and models can expand the possibilities of  graph analysis beyond 
 the traditional graph-based computation, like shortest-path, centrality analysis and so on.
 In spite of the above,   not many proposals   have been presented in this 
 sense so far. In addition, most of the existing work 
 addresses homogeneous  graphs (that is, graphs where all  nodes are of the same type), 
 where the measure of interest is related  to the OLAP analysis  on the graph 
 topology~\cite{graphOLAP,Wang2014,graphCube}. Further, existing works only address 
 graphs with binary relationships (see  Section~\ref{sec:related} for an in-depth discussion on these issues). 
 However, real-world graphs 
 are complex and often heterogeneous, where nodes and edges can be of different  
 types, and relating different numbers of entities.  
 
This paper proposes a MD data model for graph analysis, that considers not only the basic graph data, but 
background information in the form of dimension hierarchies  as well. The graphs in this  model are node- and
edge-labelled directed multi-hypergraphs, called \emph{graphoids}. In  essence, these can be denoted ``property 
hypergraphs''. A graphoid can be defined at several different levels of granularity, using the dimensions associated 
with them. For this, the {\sf Climb} operation is available. Over this model, operations like 
the ones used in typical OLAP on cubes are defined, namely {\sf Roll-Up}, {\sf Drill-Down}, {\sf Slice}, and {\sf Dice}, as
 well as other operations for graphoid manipulation, e.g., {\sf n-delete} (which deletes nodes).
The  hypergraph model allows a  natural representation of facts with different dimensions, since hyperedges can connect a variable number of nodes of different types. A typical example is the analysis of  phone calls,  the running example that will be used throughout this paper. Here, not only point-to-point calls between two partners can be represented, but also ``group calls'' between any number of participants.   In classic OLAP~\cite{Kimball1996},  a group call must  be represented by means of  a fact table containing a fixed number of columns (e.g., caller, callee, and the corresponding measures). Therefore, when the OLAP analysis for telecommunication information concerns point-to-point calls between two partners, the relational representation (denoted  ROLAP) works fine, but when this is not the case, modelling and querying issues appear, which calls for a more natural representation, closer to the original data format. And here is where the hypergraph model comes to the rescue~\cite{GomezKV17}. 
In summary, the  main contributions of the paper  are:

 \begin{enumerate}
\item A graph data model based on the notion of graphoids;
\item The definition of a collection of OLAP operations over these graphoids;
\item A proof that the classical OLAP operations on cubes can be simulated by  the OLAP operations defined in the graphoid model and,  therefore, that these graphoid-based operations are at least as powerful as the  classical OLAP operations on cubes;
\item A case study and  a series of experiments,  that give the intuition of a class of problems where the graphoid model works  clearly better than relational OLAP, whereas  for  classic OLAP queries, the graph representation is still competitive with the relational alternative.  
\end{enumerate}

In addition to the above, of course  all the classic  analysis  tools from graph theory are supported by the model, although this topic is
beyond the scope of  this paper.

\begin{remark} This paper does not claim that the graphoid model  is always more appropriate than  the classic relational OLAP representation.  Instead, the proposal  aims at showing  that when   a more flexible model is needed, where \textit{n-ary} relationships between instances are present (and \textit{n} is variable), the model allows  not only for  a more natural representation, but also can  deliver better performance for some critical queries. \qed
\end{remark}

The remainder of this paper is organized as follows:  Section~\ref{sec:related} discusses related work. Section~\ref{sec:datamodel}
presents the graphoid data model. Section~\ref{sec:olap-operations} presents the OLAP operations on graphoids, while 
  Section~\ref{sec:classical-olap}  shows that the graphoid OLAP operations capture the 
classic OLAP operations on cubes. Section~\ref{sec:casestudy} discusses a case study and presents an experimental analysis.  
  Section~\ref{sec:conclu} concludes the paper.

 \section{Related Work}
  \label{sec:related}
  
 The model described in the next sections is based on the notion of 
 property graphs~\cite{AnglesABHRV17}. In this model, nodes and edges (hyperdeges, as will be  explained later) are labelled with a  sequence  of  attribute-value pairs. It will be assumed that the values of  the  attributes   represent members  of dimension levels (i.e., each attribute value is an element in the domain of a dimension level), and thus nodes and edges  can be  aggregated, provided  that an attribute hierarchy is defined over those dimensions.  Property graphs are the usual choice in modern graph database models used in practical implementations. Attributes are included in nodes and edges mainly aimed at  improving  the speed of retrieval of  the data directly related  to  a  given  node. Here,  these attributes are also used to perform OLAP operations.

A key difference between  existing works, and the  proposal introduced in this paper, is that the latter supports the notion of \textit{OLAP hypergraphs}, highly expanding the possibilities of analysis. This way, 
instead of binary relationships between nodes, there are n-ary, probably duplicated relationships, which are typical in 
Data Warehousing and OLAP. Further, supporting n-ary relationships allows naturally modelling OLAP situations where 
different facts have a different number of relations, like in the group calls case commented in Section~\ref{sec:introduction}, and  studied in Section~\ref{sec:casestudy}.  In other words, the model handles multi-hypergraphs.  Also, the paper works over  the classic OLAP operations, and formally defines their meaning in a graph context. This approach allows an OLAP user to work  with the notion of a data cube at the conceptual level~\cite{VZ14}, regardless the kind of underlying data (in this case, graphs),  defining OLAP operations in terms of cubes and dimensions rather than in terms of nodes and edges. Finally, the authors have shown  the usefulness of this proposal in different scenarios, like trajectory analysis~\cite{GKV19} and typical OLAP analysis on social networks~\cite{VBV19}.




\section{Data Model}\label{sec:datamodel}

This section presents the graphoid OLAP data model. First,   background dimensions are formally  defined, along the lines of the classic OLAP literature. Then, the (hyper)graph data model is introduced.

\subsection{Hierarchies and  Dimensions}\label{subsec:instances}

The notions of dimension schema   and dimension graph (or dimension instance) that will be used throughout the paper, are introduced first. 
These concepts are needed to make the paper self-contained, and to understand the examples. The reader is referred to~\cite{KV17} for full details of the underlying OLAP data model.
 
\begin{definition}[Dimension Schema,  Hierarchy and Level]\rm  \label{def:dimension-schema}
Let $D$ be a name for  a dimension. A \emph{dimension \ab schema $\sigma(D)$ for $D$} is a lattice (a partial order),  
with a unique top-node, called $All$ (which has only incoming edges) 
and a unique bottom-node, called $Bottom$ (which has only outgoing edges), 
such that all maximal-length paths in the graph go from $Bottom$ to $All$.
Any path from $Bottom$ to $All$ in a dimension schema $\sigma(D)$ is called a \emph{hierarchy} of $\sigma(D)$. 
Each node in a hierarchy (that is, in a dimension schema) is called a \emph{level}  of $\sigma(D)$.
\qed
\end{definition}

The running example used throughout this paper  analyses calls between customers, which belong to different companies. For this, as   background (contextual) information for the graph data representing calls (to be explained later), there is  a  Phone  dimension, 
with levels   Phone (representing the phone number), Customer,  City,  Country, and  Operator. There is also
 a   Time  dimension, with levels Date, Month, and Year.
The following examples explain this in detail. 

\begin{example}\rm  \label{ex:dimension-schema}
Figure~\ref{fig:dimension-schema} depicts the dimension schemas $\sigma(Phone)$ and $\sigma(Time),$ 
for the dimensions $\mbox{Phone}$ and  $\mbox{Time}$, respectively. In addition, there is also a dimension denoted $\mbox{Id}$, representing identifiers, that will be explained later.
In the dimension $\mbox{Phone}$,  it holds that $Bottom=\mbox{Phone}$, and there are two  hierarchies  denoted, respectively, as 
$$\mbox{Phone} \rightarrow \mbox{Customer} \rightarrow \mbox{City} \rightarrow \mbox{Country}\rightarrow All,$$ and $$\mbox{Phone}\rightarrow \mbox{Operator} \rightarrow All.$$  The node $\mbox{Customer}$ is an example of a level in the first of the above  hierarchies. For the dimension $\mbox{Time}$, $Bottom=\mbox{Day}$ holds, as well as 
the hierarchy  $\mbox{Day} \rightarrow \mbox{Month} \rightarrow \mbox{Year}\rightarrow  All$. 
 \qed
\end{example}

\begin{figure}[htb]
\centering
   \centerline{\includegraphics[scale=0.5]{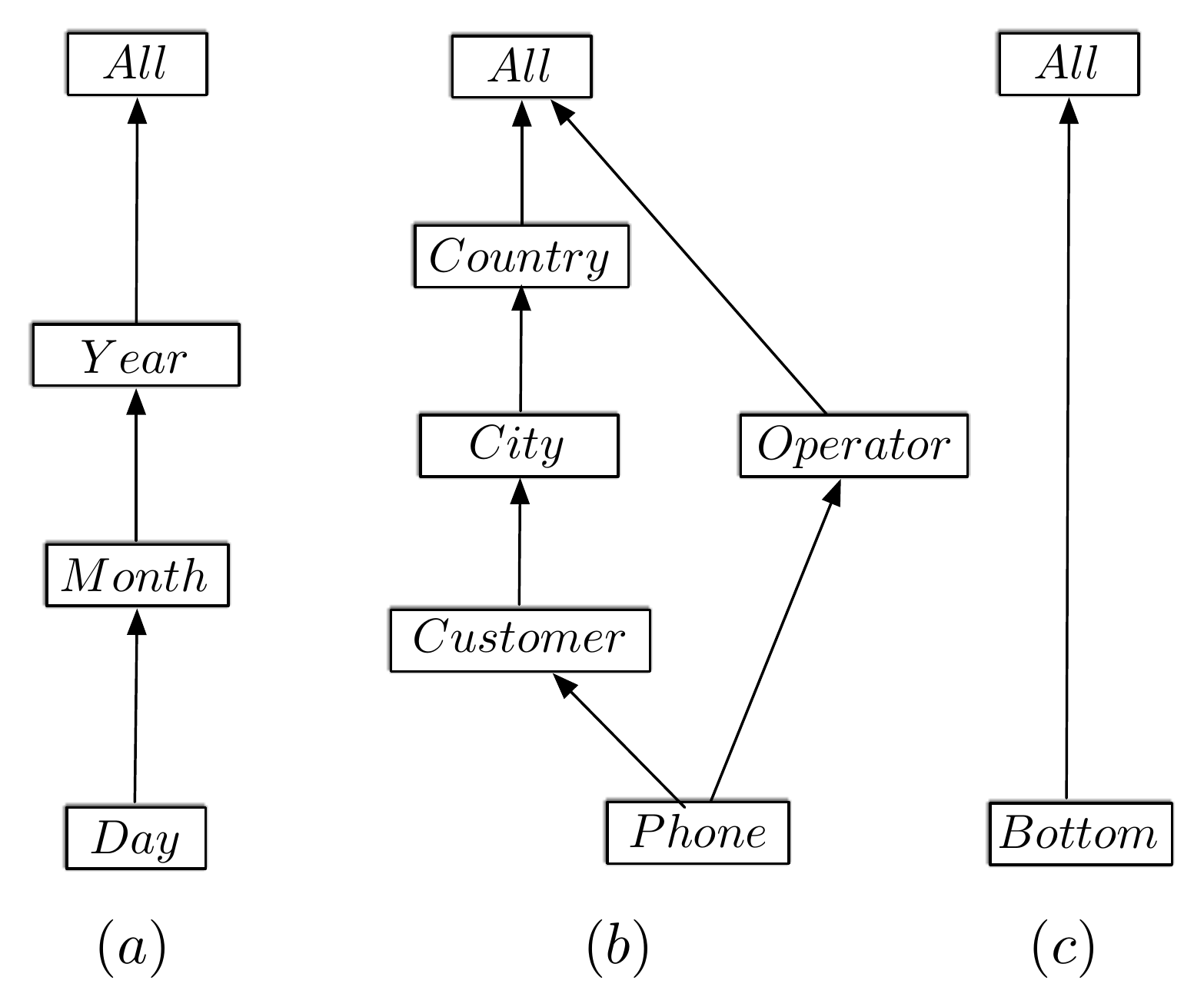}}
\caption{Dimension schemas for the dimensions $\mbox{Time}$  ($a$), $\mbox{Phone}$  ($b$), and $\mbox{Id}$ (identifier) ($c$).}\label{fig:dimension-schema}
\end{figure}

\begin{definition}[Level, Hierarchy, and Dimension Instances]\rm \label{def:instanceGraph}
Let $D$ be a dimension with schema $\sigma(D)$, 
and let $\ell$ be a level  in $\sigma(D)$. A \emph{level instance of $\ell$ } is a non-empty, finite 
set $dom(D.\ell)$. If $\ell=All$, then $dom(D.All)$ is the singleton $\{all\}$. If $\ell=Bottom$, 
then $dom(D.Bottom)$ is the  domain of the dimension $D$, 
that is,  
$dom(D)$.  

A \emph{dimension graph (or instance)} $I(\sigma(D))$ \bartb{over} the dimension schema $\sigma(D)$ is a directed acyclic 
graph with node set
$$\bigcup_{\ell} dom(D.\ell),$$ 
where the union is taken over all levels in $\sigma(D)$. The edge set of this directed acyclic graph
is defined as follows. Let $\ell$ and $\ell'$ be two levels of $\sigma(D)$, and let $a\in dom(D.\ell)$ and  $a'\in dom(D.\ell')$.
\bartb{Then, only if there is a directed edge from $\ell$ to 
$\ell'$ in $\sigma(D)$, there can be a directed edge in $I(\sigma(D))$ from $a$ to $a'$.}

If $H$ is a hierarchy in $\sigma(D)$, then the \emph{hierarchy  instance} (relative to the dimension instance $I(\sigma(D))$)
is the subgraph   of $I(\sigma(D))$ with nodes from $dom(D.\ell)$, for $\ell$ appearing in $H$. This subgraph is denoted 
 $I_H(\sigma(D))$.
\qed
\end{definition}

\begin{remark} A hierarchy instance $I_H(\sigma(D))$ is always a (directed) tree, since a hierarchy is a linear lattice. 
The following terminology is used. 
If $a$ and $b$ are two nodes in a hierarchy instance $I_H(\sigma(D))$, 
such that $(a,b)$ is in the transitive closure of the edge relation of $I_H(\sigma(D))$, 
then it is said that $a$ \emph{rolls-up} to $b$, and  denoted by $\rho_H(a,b)$ (or $\rho(a,b)$ 
 if $H$ is clear from the context). Example~\ref{ex:instanceGraph} illustrates these concepts. \qed
\end{remark}
 
\begin{example}\rm  \label{ex:instanceGraph}
Consider dimension $\mbox{Phone}$   whose schema $\sigma(Phone)$ is given in  Figure~\ref{fig:dimension-schema} ($b$). Associated with this schema, there is an  
instance where  $dom(Phone)=  \ab dom(Phone.Bottom) = \ab dom(Phone.Phone)=$ $ \{Ph_1, Ph_2, $  $ Ph_3, Ph_4,  Ph_5\}$. 
Also, at the $\mbox{Operator}$  level,  $dom(Phone.Operator)=\ab \{ATT,$ $Movistar, \ \ab Vodafone\}$.
This dimension instance $I(\sigma(Phone))$ is depicted in Figure~\ref{fig:dimension-instance}, which shows, e.g., that  phone 
 lines $Ph_2$ and  $\ Ph_4$ correspond to the operator  $Vodafone$.
\qed
\end{example}

\begin{figure}[htb]
\centering
   \centerline{\includegraphics[scale=0.5]{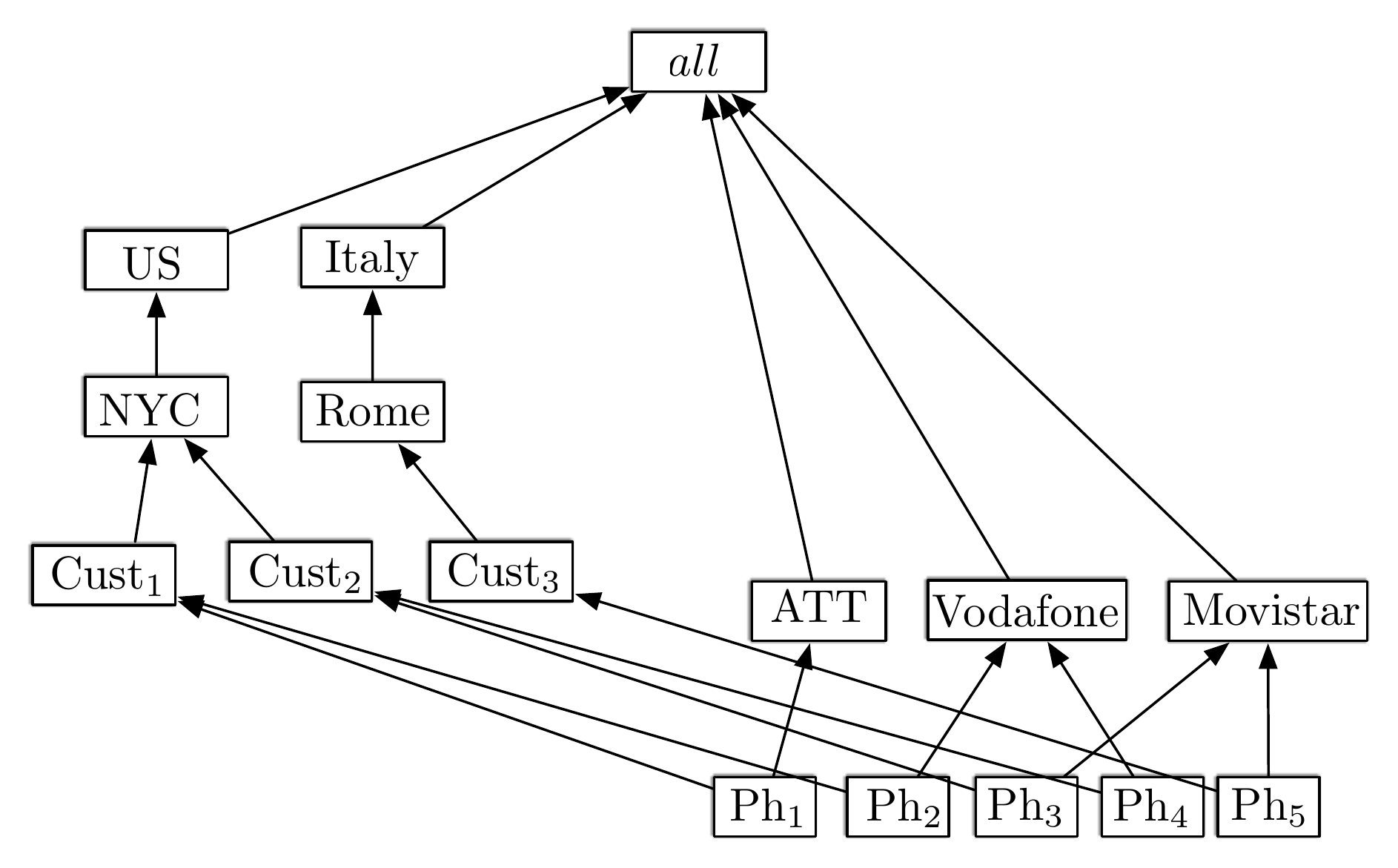}}
\caption{An example of a dimension instance  $I(\sigma(Phone))$ for the dimension $Phone$.}\label{fig:dimension-instance}
\end{figure}

In what follows,  ``sound'' dimension graphs are assumed. In thses graphs,   rolling-up from the $Bottom$  level, to the same element along different paths,  gives the same result~\cite{KV17}, typical in so-called balanced (or homogeneous) dimensions~\cite{VZ14}.
             
\subsection{The Base Graph and Graphoids}\label{subsec:graphoids}

As a basic data structure for modelling OLAP on graph data,  the concept of \emph{graphoid}  is introduced and  defined
in this section. A graphoid plays the role of a multi-dimensional cuboid in classical OLAP and it  is designed to represent  
the information of the application domain, at a certain level of granularity. 
Essentially, \textit{a graphoid is a node- and edge-labelled directed multi-hypergraph.}

In what follows, a collection of dimensions $D_1,..., D_d$ is assumed in the application domain, and  their schemas $\sigma(D_1),\ab ..., \ab \sigma(D_d)$ are given. Furthermore,  hierarchy instances  $I(\sigma(D_1)),\ab ..., \ab I(\sigma(D_d))$ for all dimensions are given. Finally,  assume that a special dimension $D_0=\mbox{Id}$ is given, to represent  unique identifiers  (Figure~\ref{fig:dimension-schema}(c)).    The notions of  \emph{attributes}, \emph{node types} and \emph{edge types} are defined next. 

\paragraph{Attributes}
The set of attributes $\att$ that describe the data is defined as 
$\att=\{D.\ell\mid D\in\{D_0,D_1,..., D_d\}\mbox{ and } \ab \ell \mbox{ is a level of } D\}.$ 
As described in Section~\ref{subsec:instances}, to each attribute $A$ of $\att$, a \emph{domain} $dom(A)$ is associated, from which the attribute  takes values. 
 
\paragraph{Node types}
Assume a finite, non-empty set $\nn$  of \emph{node types}. Elements of $\nn$ are denoted by a string starting with a hashtag. For example, the node type $\#\mbox{Phone}$ indicates that a node in a graph represents a phone line number.
There are also two functions, $ar$ and $dim$ defined    on $\nn$.
 For  each node type $\#\mbox{n}$ in $\nn$, $ar(\#\mbox{n})$ is a natural number, called the \emph{arity}, 
that expresses the number of attributes  associated with a node of type $\#\mbox{n}$.  Also, $dim(\#\mbox{n})$ is an $ar(\#\mbox{n})$-tuple of attributes, which are 
dimensions defined at the $Bottom$ level, the first of which is the Identifier dimension. 
This means that $dim(\#\mbox{n})$ is an element of $\{\mbox{Id}\}\times \{D_1,..., D_d\}^{ar(\#\mbox{n})-1}$. 
The tuple $dim(\#\mbox{n})$ tells which attributes are associated with a node of type $\#\mbox{n}$, without specifying their levels. 
Finally, assume that $dim(\#\mbox{n})$ contains no repetition, which is the usual case in practice. 
The identifier dimension is always used at its \textit{Bottom} level.

\paragraph{Edge types}
Assume the existence of a finite, non-empty set $\ee$  of \emph{edge types}, which is disjoint from the set $\nn$. Elements of $\ee$ will also be denoted by a string starting with a hashtag. For example, the node type $\#\mbox{{\sf Call}}$ indicates that an edge connects nodes that participate in a call.
Again, also assume the existence of the functions $ar$ and $dim$ on $\ee$.
To each edge type $\#\mbox{e}$ in $\nn$, $ar(\#\mbox{e})$ is a natural number, called the \emph{arity}, that expresses the number of attributes associated with an edge  of type $\#\mbox{e}$. Also, $dim(\#\mbox{e})$ is an 
$ar(\#\mbox{e})$-tuple of attributes, which are 
dimensions (at Bottom-level). 
This means that $dim(\#\mbox{e})$ is an element of $   \{D_0,D_1,..., D_d\}^{ar(\#\mbox{e})}$. 
The tuple $dim(\#\mbox{n})$ expresses which attributes are associated with an edge of type $\#\mbox{e}$, without specifying their levels.  Finally,  assume that $dim(\#\mbox{e})$ contains no repetition. 
The identifier dimension (at its \textit{Bottom}  level) may appear, but is not required.
If the identifier dimension appears, this only occurs once, among the attributes that describe edges of a certain type. 
  
It is now possible to define the notion of \textit{graphoid}.
 
\begin{definition}[Graphoid]
\label{def:graphoid}\rm
Let $D_0=\mbox{Id}$ be the identifier dimension. 
Let dimensions $D_1,..., D_d$ be given with their respective schemas and instances.
Let $\ell_1,..., \ell_d$ be levels for these respective dimensions. A \emph{$(D_1.\ell_1,..., D_d.\ell_d)$-graphoid} (or \emph{graphoid}, for short, if the levels are clear from the context)
is a 6-tuple $G=(N,\tau_N,\lambda_N,$ $E, \tau_E, \lambda_E)$, where 
\begin{itemize}
\item $N$ is a finite, non-empty set, called  the set of \emph{nodes} of $G$;
\item $\tau_N$ is a function from $N$ to $\nn$ (that associates a unique type with each node of $G$);
\item $\lambda_N$ is a function 
that maps a node $n\in N$ to a string $[\#\mbox{n}, a_1,..., a_{ar(\#\mbox{n})}]$, where $\#\mbox{n}=\tau_N(n)$ and, 
if $dim(\#\mbox{n})=(A_1,..., A_{ar(\#\mbox{n})}) $, then, for $i=1,..., ar(\#\mbox{n})$,  
$a_i\in dom(D_j.\ell_j)$, if  $A_i$ is the dimension $D_j$. 
It is assumed that  different $a_1$-values are associated with different nodes, since the first attribute value acts as a node identifier; $\lambda_N$ is denoted the \emph{node labelling function};
\item  $E$ is a subbag\footnote{Let $A$ and $B$ be bags (or sets). 
If the number of occurrences of each element 
$a$ in $A$ is less than or equal to
the number of occurrences of $a$ in $B$, then $A$ is called a \emph{subbag} of $B$, also denoted $A\subseteq B$.} of the set ${\cal P}(N)\times {\cal P}(N)$, which we call the 
set of \emph{(multi hyper-)edges} of $G$;  
\item $\tau_E$ is a function from $E$ to $\ee$ (that associates a unique type to each edge of $G$); and 
\item $\lambda_E$  is a function 
that maps a hyperedge $e\in E$ to a string $[\#\mbox{e}, \ab b_1,\ab ...,\ab  b_{ar(\#\mbox{n})}]$, where $\#\mbox{e}=\tau_E(e)$ and, if $dim(\#\mbox{e})=(B_1,..., B_{ar(\#\mbox{n})}) $, then, for $i=1,..., ar(\#\mbox{e})$,  
$b_i\in dom(D_j.\ell_j)$, if  $B_i$ is the dimension $D_j$; 
  $\lambda_E$ is called the \emph{edge labelling function}.
\qed
\end{itemize}
\end{definition}

The basic graph data that serves as input data to the  graph OLAP process, is called the \emph{base graph}. A base graph plays the role of a multi-dimensional cube in classical OLAP and is designed to contain all the information of the application domain, at the lowest level of granularity.

 \begin{definition}[Base graph]
\label{def:basegraph}\rm
Let dimensions $D_1,..., D_d$ be given with their respective schemas and instances.
The $(D_1.\mbox{Bottom},...,$  $D_d.\mbox{Bottom})$-graphoid  is called the \emph{base graph}.
\qed
\end{definition}

\medskip

 \begin{example}\label{ex:base-graph}\rm
 The running   example  used in this paper is aim\-ed at  analysing calls between customers of phone lines; lines correspond to different operators. Examples~\ref{ex:dimension-schema} and~\ref{ex:instanceGraph}  showed some of the dimensions  used as background information. Next,  the call information is  shown, represented as a graph. The  \textsf{Phone} dimension plays the roles of the calling line and the callee lines (this is called a role-playing dimension in the OLAP literature~\cite{VZ14}). The information in the hyperedges  reflects the total duration of the calls between two or more phone numbers on a given day.
Figure~\ref{fig:fig-phone-as-graph} shows an example of a base graph, where   $N=\{1,2,3,4,5\}$ is the  node set. 
 The nodes in this base graph are all of the same type and represent phones (not persons--a person may have more than one phone). 
In this example, $\nn=\{\#\mbox{{\sf Phone}}\}$. The node type $\#\mbox{{\sf Phone}}$  has arity $2$. Its first attribute is a node identifier and the second one   is a dimensional attribute that represents the phone number, with domain $\{\mbox{Ph}_1, \mbox{Ph}_2, ...\}$. In the example of Figure~\ref{fig:fig-phone-as-graph},  
$$\lambda_N : i\mapsto [\#\mbox{{\sf Phone}},10+i, \mbox{Ph}_i], \mbox{for}~ i=1,...,5.$$

 Hyperedges represent phone calls, which most of the time 
 involve two phones, but which may also involve multiple phones, representing so-called ``group calls.'' 
 So, edges are all of the same type $\#\mbox{{\sf Call}}$ and  $\ee=\{\#\mbox{{\sf Call}}\}$.
 In Figure~\ref{fig:fig-phone-as-graph}, a directed  hyperedge from a subset $S$ of $N$ to a subset $T$ of $N$ is 
 graphically represented by a coloured node which has incoming arrows (of the same colour) from all elements of $S$ and outgoing  
 arrows (again of the same colour) to all elements of $T$. Such a coloured construction is a 
 depiction of the hyperedge $e=(S,T)$, which will be denoted 
 $S\rightarrow T$ from now on.\footnote{The nodes of $S$ are called the \emph{source nodes} of $e$ and the nodes of $T$ are called the \emph{target nodes} of $e$. 
The source and target nodes of $e$ are called \emph{adjacent} to $e$, and   the set of the adjacent nodes to $e$ 
 is denoted by $Adj(e)$. Thus,  $Adj(e)=S\cup T$.}
 For example, the red and purple hyperedges $\{1\}\rightarrow\{2\}$ represent two different phone calls from $\mbox{Ph}_1$ to $\mbox{Ph}_2$,
 made on the same day and of the same duration. This example explains why the model assumes  bags rather than sets. 
  The orange hyperedge $\{3\}\rightarrow\{2,5\}$ represents a group call, from $\mbox{Ph}_3$ to both
 $\mbox{Ph}_2$ and $\mbox{Ph}_5.$
  There are six phone calls shown in the figure. So, $E$ is the bag $\{\!\!\{  \{1\}\rightarrow\{2\}, \{1\}\rightarrow\{2\},
 \{4\}\rightarrow\{3\} , \{4\}\rightarrow\{5\}, \{3\}\rightarrow\{2,5\}, \{5\}\rightarrow\{2,3\}     \}\!\!\}.$
 The edge labelling function $\lambda_E$ associates two attributes, with edges of type $\#\mbox{{\sf Call}}$, namely Date and Duration.
 Date  is a dimensional attribute to which the dimensional  hierarchy in Figure~\ref{fig:dimension-schema}  is associated. Duration is a measure
 attribute (which has as an associated aggregation function, in this case, the summation). 

 \begin{figure}[h]
            \centering
            \includegraphics[scale=0.7]{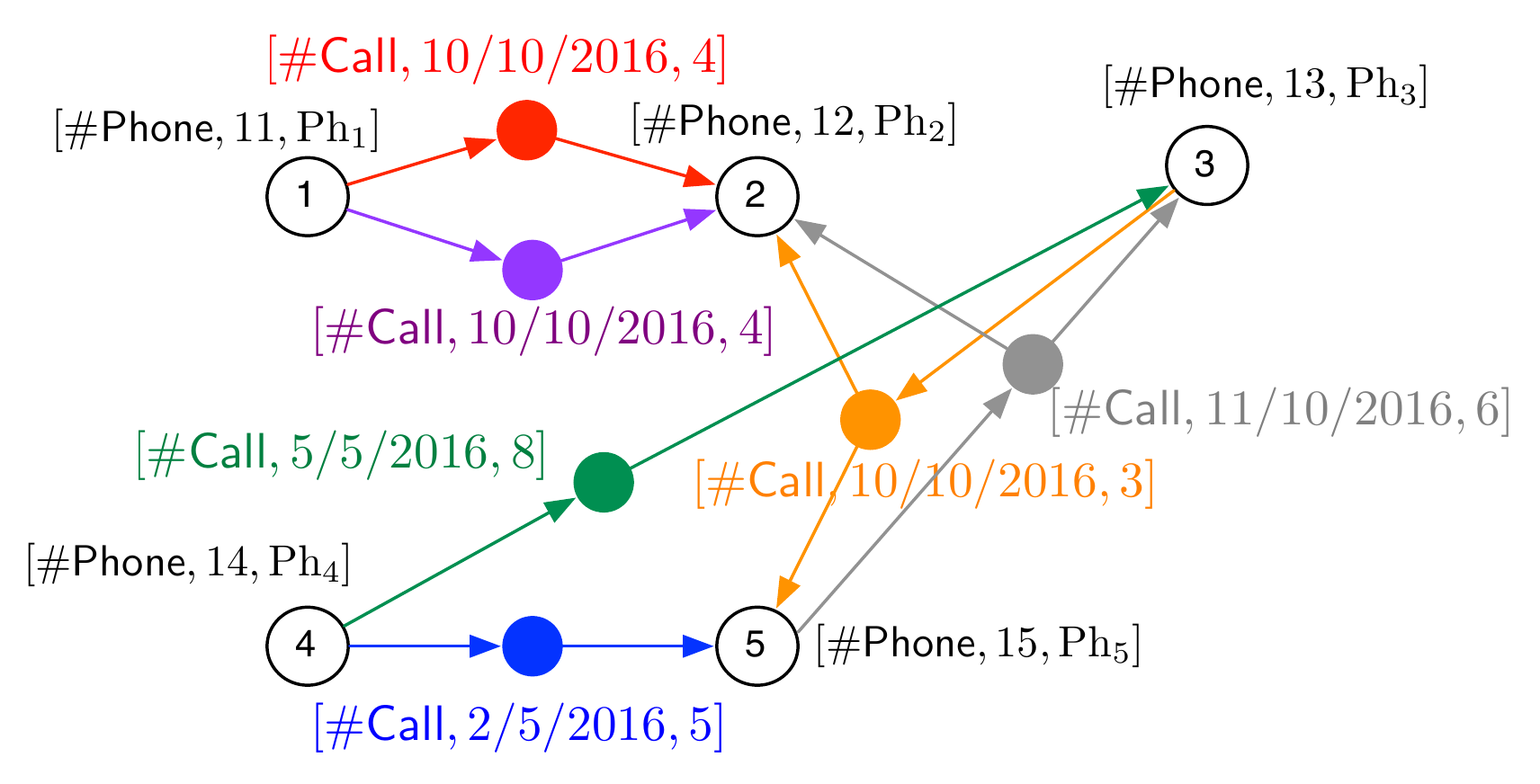}
            \caption{Basic phone call data as a base graph.}
            \label{fig:fig-phone-as-graph}
 \end{figure}  
  \qed
 \end{example}
         
Note that, although the  base graph plays the role of a multi-dimensional cube in classical OLAP
(or a fact table in relational OLAP), a key difference is that this cube has a variable number of ``axes'', since it can represent facts including a variable number of dimensions.
The next example  discusses two  graphoids whose dimensions are at    different levels of granularity.
Later  it will be  explained how these graphoids can be obtained from the base one.
         
 \begin{example} \label{ex:graphoid}\rm
 Continuing with Example~\ref{ex:base-graph}, consider  two available dimensions, namely 
 $D_1=\mbox{Time} $ and $D_2=\mbox{Phone}$. A 
$(\mbox{Time.Day}, \ab \mbox{Phone.Operator})$-graphoid 
 based on the base graph of Figure~\ref{fig:fig-phone-as-graph}, is shown in Figure~\ref{fig:fig-phone-RU-company-alternative}. Here, in the Phone nodes, the phone numbers   have been replaced with their corresponding operator name, at the Phone.Operator level in the dimension $\mbox{Phone}$ (e.g., for $\mbox{Ph}_3$, the corresponding operator is Movistar).

   \begin{figure}[h]
            \centering
            \includegraphics[scale=0.7]{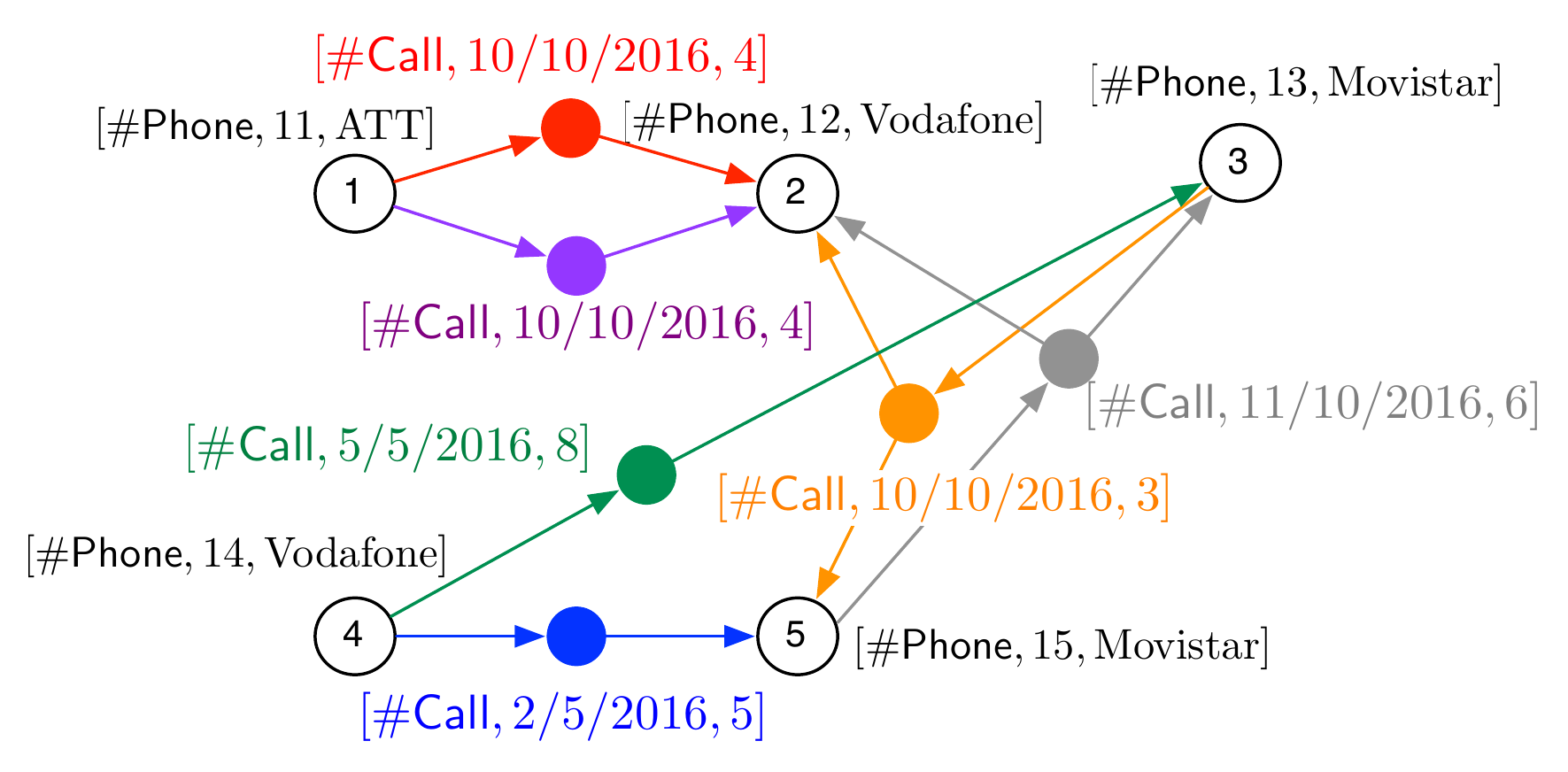}
            \caption{A  $(\mbox{Time.Day}, \mbox{Phone.Operator})$-graphoid, based on the data shown in Figure~\ref{fig:fig-phone-as-graph}.}
            \label{fig:fig-phone-RU-company-alternative}
 \end{figure}

  Figure~\ref{fig:fig-phone-RU-company} shows an alternative $(\mbox{Time.Day}, \mbox{Phone.Operator})$-graphoid  for the data from Figure~\ref{fig:fig-phone-as-graph}. This graphoid has $N=\{1,2,3\}$ as a node set. The nodes with identifiers 12 and 14 represent, respectively,  $\mbox{Ph}_2$ and $\mbox{Ph}_4$ in the base graph (and also in the 
graphoid of Figure \ref{fig:fig-phone-RU-company-alternative}),  which belong to the operator Vodafone. Thus, these two nodes
were collapsed  into one  (with identifier 12)
and similarly, the nodes $\mbox{Ph}_3$ and 
$\mbox{Ph}_5$ were collapsed into one node (with identifier 13). 
These operations were  possible because  these nodes have identical attribute values (apart from the identifier).
 For the dimension $\mbox{Time}$, all information in  Figure~\ref{fig:fig-phone-RU-company}
 is at the level of $\mbox{Day}$ and all information for the dimension $\mbox{Phone}$ is at the level of $\mbox{Company}$. 
 These examples show that there can be more than one $(\mbox{Time.Day},\ab \mbox{Phone.Operator})$-graphoids   ``consistent''   with the given base graph. Thus, some kind of normalization is needed. This is studied in the next section.
  \qed
 \end{example}
 
  \begin{figure}[h]
            \centering
            \includegraphics[scale=0.7]{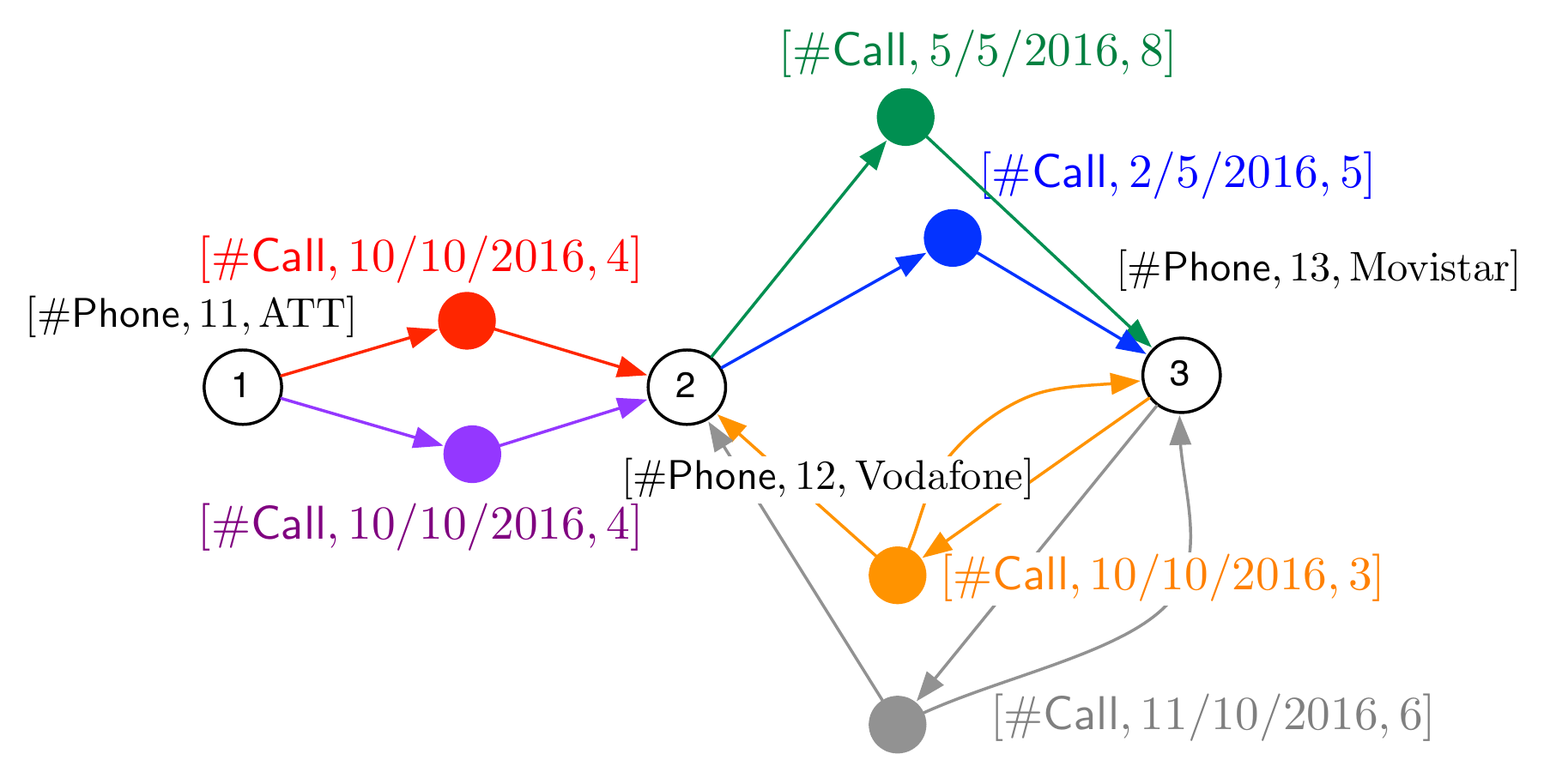}
            \caption{An alternative $(\mbox{Time.Day}, \mbox{Phone.Operator})$-graphoid, based on the data shown in Figure~\ref{fig:fig-phone-as-graph}.}
            \label{fig:fig-phone-RU-company}
 \end{figure}

\begin{remark}\rm\label{remark:conceptual}
Nodes are assumed to represent basic objects in the modelled  application world.  
These objects are given by a number of descriptive attributes. 
Measure  information,  typically present in an OLAP setting  to quantify facts, is, in this  philosophy, represented as attributes on the hyperedges. The call duration is an example of a measure that is placed on edges of the type  $\mbox{\sf Call}$.
However, the above definition also allows for node attributes to be dimensions
that contain measure information. Consider a slightly modified situation in which an object of
 type $\#\mbox{\sf Phone}$ includes an
additional attribute that expresses the average (or expected) billing amount for that particular phone number, for example, 
$[\#\mbox{\sf Phone}, 11, \mbox{Ph}_1, 880]$. In this modified setting, a user may want to compute  
the average expected billing amount over all phone lines. To answer these kinds of   queries,  
attribute values of certain types of  nodes must be averaged (in the example, the $\#\mbox{\sf HasExpectedBill}$ attribute). 
However, in the model presented here, \textit{aggregations are only performed on attribute values of hyperedges.} 
Whenever this problem occurs,  the representation can be modified as 
 illustrated in Figure~\ref{fig:fig-edgify}. On the left-hand side, there is a node that includes
 the $\#\mbox{\sf HasExpectedBill}$ attribute. On the right-hand side, this attribute is brought to the 
$All$ level in its dimension and gets the value $all$. The expected billing information is moved to a new edge of type $\#\mbox{\sf HasExpectedBill}$, where it can be subject to aggregation. The above  operation is called 
the  \emph{edgification} of an attribute $A$ in a node of type $\#\mbox{\sf n}$, and it is  denoted
by $\mbox{\sf Edgify}(\#\mbox{\sf n}, A)$. \qed
  \begin{figure}[h]
            \centering
            \includegraphics[scale=0.7]{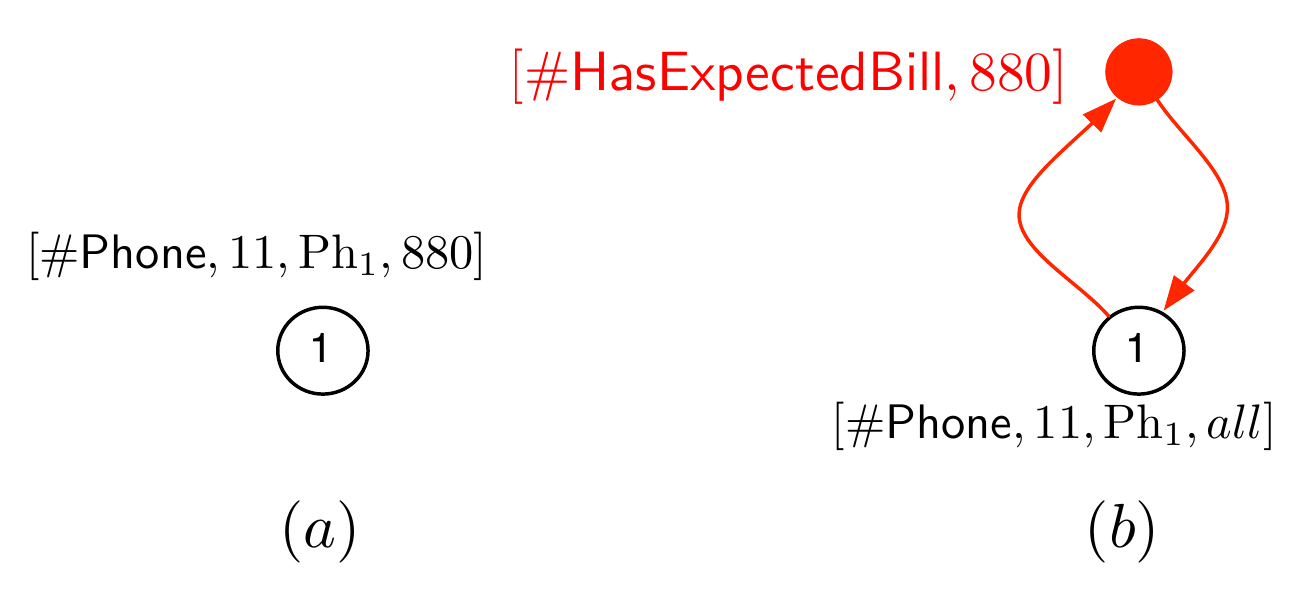}
            \caption{(a) A node with label $[\#\mbox{\sf Phone}, 11, \mbox{Ph}_1, 880]$, where 880 expresses the expected bill.
             (b) An edgification of this node, where the expected billing information is moved to an edge that is  labelled  $\#\mbox{\sf HasExpectedBill}$.}
            \label{fig:fig-edgify}
 \end{figure} 

\ignore{
This is consistent with both possibilities of data acquisition and modelling, namely ETL and ELT, 
commented in Section~\ref{sec:introduction}. In the latter, mostly used in a ``Big Data'' scenario, it is usual to capture data as they come, and then decide whether or not to model these data in a particular way. Edgification is useful in these scenarios. When it is possible 
  to model data prior to exploitation, the choice would most likely be  to represent measures as edges. However, when data are acquired    as they come, it is most likely that measures like the ones above, come as attribute of some node or dimension. Then, the notion of edgification gives the possibility of changing this situation on-the-fly. }
\end{remark}

\subsection{Minimal graphoids }\label{subsec:minimal-graphoids}
 In this section,   the notion of 
\emph{minimal} $(D_1.\ell_1,\ab ..., \ab D_d.\ell_d)$-graphoid is defined. This graphoid is 
obtained collapsing the nodes that have identical labels (apart from the identifier) in the original  graphoid. 
Let $G=(N,\tau_N,\lambda_N, E, \tau_E, \lambda_E)$ be a  
$(D_1.\ell_1,..., D_d.\ell_d)$-graphoid. If the nodes  $n_1, n_2\in N$ 
have identical labels, apart from the identifier, denoted $\lambda_N(n_1)=_{\mbox{\tiny Id}}\lambda_N(n_2)$,  then  these nodes are identified, such that only the one with the  smallest identifier is preserved,
while  the others are deleted. So, if the $\lambda_N$-values of the nodes $n_1,n_2,..., n_k$ 
 pairwise satisfy the $=_{\mbox{\tiny Id}}$-relationship, and 
$n_1$ has the smallest identifier among them, then   the nodes $n_2,..., n_k$ are  replaced  by 
$n_1$ and then deleted.  The expression  $rep_N(n_i)=n_1$, 
for $i=1,2,..., k$, indicates that $n_1$ represents the nodes $n_1,n_2,..., n_k$ in the minimal graph.
All edges leaving from or arriving at the nodes $n_2,..., n_k$ are redirected to $n_1$.
For this purpose,   the function $rep_N$ is defined  on subsets of the node set $N$: if $S\subseteq N$, then $rep_N(S)= \{rep_N(n)\mid n\in S\}$. Now,  the notion of minimal graphoid is defined more formally.

\begin{definition}[Minimal graphoid]
\label{def:minimal-graphoid}\rm
Let $D_0,  D_1,..., D_d$  and $\ell_1,..., \ell_d$ be the same as  in Definition~\ref{def:graphoid}. 
Let  $G=(N,\tau_N,\lambda_N, E, \tau_E, \lambda_E)$ be a $(D_1.\ell_1,..., D_d.\ell_d)$-graphoid.   
The \emph{minimal graphoid of} $G$ is the $(D_1.\ell_1,..., D_d.\ell_d)$-graphoid  
 $G'=(N',\tau_{N'},\lambda_{N'}, E', \tau_{E'}, \lambda_{E'})$, defined as follows:

\begin{itemize}
\item $N'$ is the set $rep_N(N)=\{ rep_N(n)\mid n\in N\}$; 

\item $\tau_{N'}$ is a function from $N'$ to $\nn$, defined as $\tau_{N'}(rep_N(n)):=\tau_N(rep_N(n))$, for each $n$ in $N$;  

\item $\lambda_{N'}$ is a function on $N'$  defined as  $\lambda_{N'}(rep_N(n)):=\lambda_N(rep_N(n))$, for each $n$ in $N$;  
 
\item  $E'$ is a subbag
of the set ${\cal P}(N')\times {\cal P}(N')$, defined as follows: for each hyperedge $e=S\rightarrow T$ in  $E$, then a new hyperedge $rep_N(e):=rep_N(S)\rightarrow rep_N(T)$ is in  $E'$;

\item $\tau_{E'}$ is a function from $E'$ to $\ee$, defined as $\tau_{E'}(rep_N(e)):=\tau_E(e)$, for each $e$ in $E$;  

\item $\lambda_{E'}$ is a function on $E'$ and it is defined as  $\lambda_{E'}(rep_N(e)):=\lambda_E(e)$, for each $e$ in $E$.
 
\qed
\end{itemize}
\end{definition}

\begin{remark} 
The set $N$ of nodes of $G$ is contracted to the set $N'=rep_N(N)$,  therefore each node in $N'$ has the smallest identifier among all nodes that are mapped to $n$ by the $rep_N$-function.  For edges, $E'$ is defined as the bag $\{\!\!\{ rep_N(e)\mid e\in E\}\!\!\}$, which means that for each hyperedge in $E$, there is a corresponding hyperedge in $E'$. This means that the cardinalities of the bags $E$ and $E'$ are the same. \qed
\end{remark}

Proposition~\ref{prop:unique-minimal-graphoid} immediately follows from Definition~\ref{def:minimal-graphoid}.
 
 \begin{proposition}\rm\label{prop:unique-minimal-graphoid}
 For any $(D_1.\ell_1,..., D_d.\ell_d)$-graphoid
 $G=(N,\tau_N,\ab \lambda_N, \ab E, \ab \tau_E, \ab \lambda_E)$, its minimal 
$(D_1.\ell_1,\ab \ab..., D_d.\ell_d)$-graphoid always exists and it is unique.
 \qed
 \end{proposition} 
 
 \begin{example} \label{ex:mingraphoid}\rm
 The two
$(\mbox{Time.Day},\ab \mbox{Phone.Operator})$-graphoids 
 shown in Figures~\ref{fig:fig-phone-RU-company-alternative} and~\ref{fig:fig-phone-RU-company}  in  Example~\ref{ex:graphoid}, correspond to the   graph of Figure~\ref{fig:fig-phone-as-graph}.
The graphoid of Figure~\ref{fig:fig-phone-RU-company} is the minimal 
graphoid of Figure~\ref{fig:fig-phone-RU-company-alternative}.
In this example, the original nodes $2$ and $4$ are contracted into one node, namely the node $2$ (since it has the smallest identifier of the two).
Similarly, the original nodes $3$ and $5$ are contracted into  the node $3$.
The original node $1$ remains unchanged. Between  nodes $1$ and $2$, there are two edges (with the same label) in the original graph. They are copied in the minimal graph.
The edges between   nodes $4$ and $3$, and $4$ and $5$, respectively, become two edges between the nodes $2$ and $3$
in the minimal graph. The two hyperedges that involve   nodes $2$, $3$ and $5$ correspond to two hyperedges between the nodes $2$ and $3$ in the minimal graph.
   \qed
 \end{example}
 
 For any $(D_1.\ell_1,..., D_d.\ell_d)$-graphoid
 $G=(N,\tau_N,\lambda_N, E, \tau_E, \lambda_E)$, 
 the result of the minimisation described in this section is denoted   $\mbox{\sf Minimize}(G)$, and    
 called the \emph{minimisation} of $G$. 
 
 \begin{remark}\rm\label{remark:compute-minimal}
 It is easy to see that the minimal $(D_1.\ell_1,..., D_d.\ell_d)$-graphoid  of  a  $(D_1.\ell_1,..., D_d.\ell_d)$-graphoid
 $G=(N,\tau_N,\lambda_N, E, \tau_E, \lambda_E)$ can be computed, in the worst case, in time that is quadratic
  in $|N|$ and linear in $|E|$. This can be improved,  for instance, with an early pruning of the nodes 
  that will not be contracted.  Addressing this issue is beyond the scope of this paper. \qed
 \end{remark}

\section{OLAP Operations on Graphs}
\label{sec:olap-operations}   

In this section,  the operations that compose the graph-OLAP language  over graphoids are defined.  Section~\ref{sec:classical-olap} will show that these  operations can simulate the typical OLAP operations on cubes.
  
\subsection{Climb}\label{subsec:climb} 

The $\mbox{\sf Climb}$-operation, intuitively,  allows to define graphs at different levels of granularity, based on the background dimensions. 

\begin{definition}[Climb]\rm\label{def:climb}   Assume a $(D_1.\ell_1,..., D_d.\ell_d)$-graphoid $G$ is given as follows:   $G=(N,\tau_N,\lambda_N, E, \tau_E, \lambda_E)$. Let $D_k$ be a dimension that appears in $G$, and $\ell_k$ and $\ell'_k$ be levels in the schema $\sigma(D_k)$ of this dimension, such that $\ell_k\rightarrow \ell'_k$. Also, let $\rho_{\ell_k\rightarrow\ell'_k}$ be the corresponding rollup function (at the instance level). Finally, let $\#\mbox{\sf n}$ be a node type that appears in $G$, and  $\#\mbox{\sf e}$ be an edge type that appears in $G$.

The \emph{node-climb-operation of $G$ along the dimension $D_k$ from level $\ell_k$ to level $\ell'_k$ in all nodes of type 
$\#\mbox{\sf n}$}, denoted $\mbox{\sf Climb}(G, \#\mbox{\sf n}, D_k.(\ell_k\rightarrow\ell'_k))$,  
replaces all 
attribute values $a$ from  $dom(D_k.\ell_k)$ by the  value 
$\rho_{\ell_k\rightarrow\ell'_k}(a)$ from $dom(D_k.\ell'_k)$, in all nodes of $G$ of type $\#\mbox{\sf n}$,  leaving $G$ unaltered otherwise.

The \emph{edge-climb-operation of $G$ along the dimension $D_k$ from level $\ell_k$ to level $\ell'_k$ in all hyperedges of type 
$\#\mbox{\sf e}$}, denoted $\mbox{\sf Climb}(G, \#\mbox{\sf e}, D_k.(\ell_k\rightarrow\ell'_k))$, 
replaces all 
attribute values $a$ from  $dom(D_k.\ell_k)$ by the value 
$\rho_{\ell_k\rightarrow\ell'_k}(a)$ from $dom(D_k.\ell'_k)$, in all edges of $G$ of type $\#\mbox{\sf e}$,  leaving $G$ unaltered otherwise.
\qed
\end{definition}

\begin{example}\rm\label{ex:climb}
Applying to the graphoid $G$ depicted in Figure~\ref{fig:fig-phone-as-graph} the operation $\mbox{\sf Climb}(G, \#\mbox{\sf Phone}, \mbox{Phone}.(\mbox{Phone}\rightarrow \mbox{Operator}))$, results in the graphoid shown inFigure~\ref{fig:fig-phone-RU-company-alternative}.\qed
\end{example}

\begin{remark}\rm\label{remark:climb}
If a dimension $D_k$  appears in multiple node types and edge types,  to apply the 
$\mbox{\sf Climb}$-operation on many of them, the shorthand expression
$\mbox{\sf Climb}(G, \{\#\mbox{\sf n}_1, \ab ...,\ab  \#\mbox{\sf n}_r, \ab \#\mbox{\sf e}_1,\ab ..., \ab \#\mbox{\sf e}_s\}, \ab D_k.(\ell_k\rightarrow\ell'_k))$ can be used.  
Finally, $\mbox{\sf Climb}(G, \ast , D_k.(\ell_k\rightarrow\ell'_k))$  denotes a  climbing, in the dimension $D_k$, from level $\ell_k$ to level $\ell'_k$ in all possible node and edge types.
\qed
\end{remark}

\subsection{Grouping}\label{subsec:grouping}   

The $\mbox{\sf Group}$-operation, both on nodes and on edges, is defined in this section.

\begin{definition}[Grouping]\rm\label{def:grouping}  Assume a $(D_1.\ell_1,..., D_d.\ell_d)$-graphoid $G$ is given as follows:   $G=(N,\tau_N,\lambda_N, E, \tau_E, \lambda_E)$.  Let $D_k$ be a dimension that appears in $G$ and let $\ell_k$ and $\ell'_k$ be levels in the schema $\sigma(D_k)$ of this dimension, such that $\ell_k\rightarrow \ell'_k$. Let $\rho_{\ell_k\rightarrow\ell'_k}$ be the corresponding rollup function. Let $\#\mbox{\sf n}$ be a node type that appears in $G$ and let  $\#\mbox{\sf e}$ be an edge type that
 appears in $G$.

The \emph{node-grouping of $G$ along the dimension $D_k$ from level $\ell_k$ to level $\ell'_k$ in all nodes of type 
$\#\mbox{\sf n}$}, denoted $\mbox{\sf Group}(G, \#\mbox{\sf n}, D_k.(\ell_k\rightarrow\ell'_k))$, 
is defined as $\mbox{\sf Minimize}(\mbox{\sf Climb}(G, \#\mbox{\sf n}, D_k.(\ell_k\rightarrow\ell'_k)))$.

The \emph{edge-grouping of $G$ along the dimension $D_k$ from level $\ell_k$ to level $\ell'_k$ in all hyperedges of type 
$\#\mbox{\sf e}$}, denoted $\mbox{\sf Group}(G, \#\mbox{\sf e}, D_k.(\ell_k\rightarrow\ell'_k))$, 
is defined as $\mbox{\sf Climb}(G, \#\mbox{\sf n}, D_k.(\ell_k\rightarrow\ell'_k))$.
\qed
\end{definition}

\begin{example}\rm\label{ex:group}
Applying to the graphoid $G$  depicted in Figure~\ref{fig:fig-phone-RU-company-alternative} the operation $\mbox{\sf Group}(G, \#\mbox{\sf Phone}, \mbox{Phone}.(\mbox{Phone}\rightarrow \mbox{Operator}))$, results in  the graphoid, depicted in 
Figure~\ref{fig:fig-phone-RU-company}.\qed
\end{example}

\subsection{Aggregate}\label{subsec:aggregation}   

In this section,   the $\mbox{\sf Aggr}$-operation on measures stored in edges is defined.

\begin{definition}[Aggregate]\rm\label{def:aggregate}  Given a minimal $(D_1.\ell_1,..., D_d.\ell_d)$-graphoid $G$ defined as 
 $G=(N,\tau_N,\lambda_N, E, \tau_E, \lambda_E)$, let $D_k$ be a dimension that appears in the hyperedges of $G$ of type $\#\mbox{\sf e},$ that plays the role of a measure, to which the aggregate function $F_k$ can be applied.
The  \emph{aggregation of the graphoid $G$ over the dimension $D_k$ (using the function $F_k$)}, denoted 
$\mbox{\sf Aggr}(G, \#\mbox{\sf e},  $ $D_k, F_k)$,   results in a graphoid $G'$ over the same $N, \tau_N$ and $\lambda_N$ as $G$, with the following modified hyperedge bag $E'$. If the hyperedges $e_1,e_2, ..., e_r$ are all of type $\#\mbox{\sf e}$ and all of type
$S\rightarrow T$ (and if they are the only ones), and if $\lambda_E$ agrees on all of them apart from a possible identifier-attribute, and apart from the dimension $D_k$, then the hyperedges $e_1,e_2, ..., e_r$ are replaced by one of them (say $e_1$) of the same type and with the same attribute values, apart from the identifier, which is the identifier of $e_1$, and the value of the attribute $D_k.\ell_k$, which becomes 
the value of the aggregation function $F_k$ applied to the values of the attribute $D_k.\ell_k$ of the edges $e_1,e_2, ..., e_r$.
 \qed
\end{definition}

\begin{example}\rm\label{ex:aggregate}
Applying the operation 
  $\mbox{\sf Aggr}(G, \#\mbox{\sf Call}, \mbox{Duration}, \mbox{\sc Sum})$ to the graphoid $G$, depicted in Figure~\ref{fig:fig-phone-RU-company},
 results in a graphoid where the two edges that connect the nodes $1$ and $2$ are replaced by one edge with label
  $[ \#\mbox{\sf Call}, 10/10/2016, 8]$, which contains, in the measure attribute,  the sum of the two durations.
   \qed
\end{example}

\begin{remark}\rm\label{remark:aggregate}
To aggregate multiple dimensions $M_1,\ab ..., \ab M_k$, using the aggregate functions  $F_1,\ab ...,\ab  F_k$  simultaneously, 
 the notation would be: $\mbox{\sf Aggr}(G, \ab \#\mbox{\sf e}, \ab \{M_1,..., M_k\}, \ab  \{F_1,..., F_k\})$.
 Also,  for simplicity, only the typical SQL aggregate functions {\sc Sum}, {\sc Max},{\sc Min} and {\sc Count} are considered.
\qed
\end{remark}

\begin{remark}\rm\label{remark:climb-group-aggr}
Although the operations $\mbox{\sf Climb}$, $\mbox{\sf Group}$, and $\mbox{\sf Aggr}$, are not present in classic relational OLAP, they are included here for several reasons: first, they can be useful when operating on graphs in practice; second, they facilitate and make it simple the definition of the Roll-up operation, that otherwise could be unnecessarily difficult to express. 
\qed
\end{remark}
 \subsection{Roll-Up}\label{subsec:rollup}   
 
 The operations defined above allow defining  
 the $\mbox{\sf Roll-Up}$-operation over dimensions and measures
  stored in edges, as explained  next. 

\begin{definition}[Roll-Up]\rm\label{def:rollup}   Assume a $(D_1.\ell_1,..., D_d.\ell_d)$-graphoid $G$ is given as follows:   $G=(N,\tau_N,\lambda_N, E, \tau_E, \lambda_E)$. 
Let $D_c$ be a dimension that appears in some nodes and/or hyperedges of $G$, that   plays the role of a climbing dimension.
Let $M_1,..., M_k$ be dimensions that appear in the hyperedges of type $\#\mbox{\sf e}$  of $G$. These dimensions play the role of measure dimensions, and it is  assumed that  aggregate functions $F_1,..., F_k$ are associated with them.
Let $\#\mbox{\sf n}_1,..., \#\mbox{\sf n}_r$ be node types appearing in $G$, and let $\#\mbox{\sf e}_1,..., \#\mbox{\sf e}_s$ be hyperedge 
types appearing in $G$.
The \emph{roll-up of $G$ over the dimensions $M_1,..., M_k$ (using the functions  $F_1,..., F_k$) in hyperedges of type
$\#\mbox{\sf e}$, and over 
the  climbing dimension $D_c$ from level $\ell_c$  to level $\ell'_c$ in nodes of types $\#\mbox{\sf n}_1,..., \#\mbox{\sf n}_r$ and edges of types $\#\mbox{\sf e}_1,...,\#\mbox{\sf e}_s$},  denoted $$\mbox{\sf Roll-Up}(G,\{\#\mbox{\sf n}_1,...,\#\mbox{\sf n}_r, \#\mbox{\sf e}_1,...,\#\mbox{\sf e}_s\}, D_c.(\ell_c\rightarrow\ell'_c);\#\mbox{\sf e},M_1,...,M_k,F_1,...,F_k),$$  is defined as 

$$\displaylines{\quad\mbox{\sf Aggr}(\mbox{\sf Minimize}(\mbox{\sf Climb}(G, \{\#\mbox{\sf n}_1,..., \#\mbox{\sf n}_r, \#\mbox{\sf e}_1,..., \#\mbox{\sf e}_s\},  \hfill{} \cr \hfill{}    D_c.(\ell_c\rightarrow\ell'_c))),\#\mbox{\sf e}, M_1,..., M_k, F_1,..., F_k).\quad}$$
 \qed
\end{definition}
  
\begin{example}\rm\label{ex:roll-up}
Applying   to the graphoid depicted in Figure~\ref{fig:fig-phone-RU-company} the operation 
$\mbox{\sf Roll-Up}(G, \{ \#\mbox{\sf Call}\}, \ab \mbox{Time}.(\mbox{Day}\rightarrow\mbox{Year}); \ab \#\mbox{\sf Call},\ab \mbox{Duration}, \ab \mbox{\sc Sum})$,
 results in the graphoid of Figure~\ref{fig:fig-phone-RU-company+Agg}.
 The minimisation step in the above implementation of the roll-up operation does nothing, in this case, since the operation is applied to a 
 minimal graphoid. \qed
\end{example}

   \begin{figure}[h]
            \centering
            \includegraphics[scale=0.7]{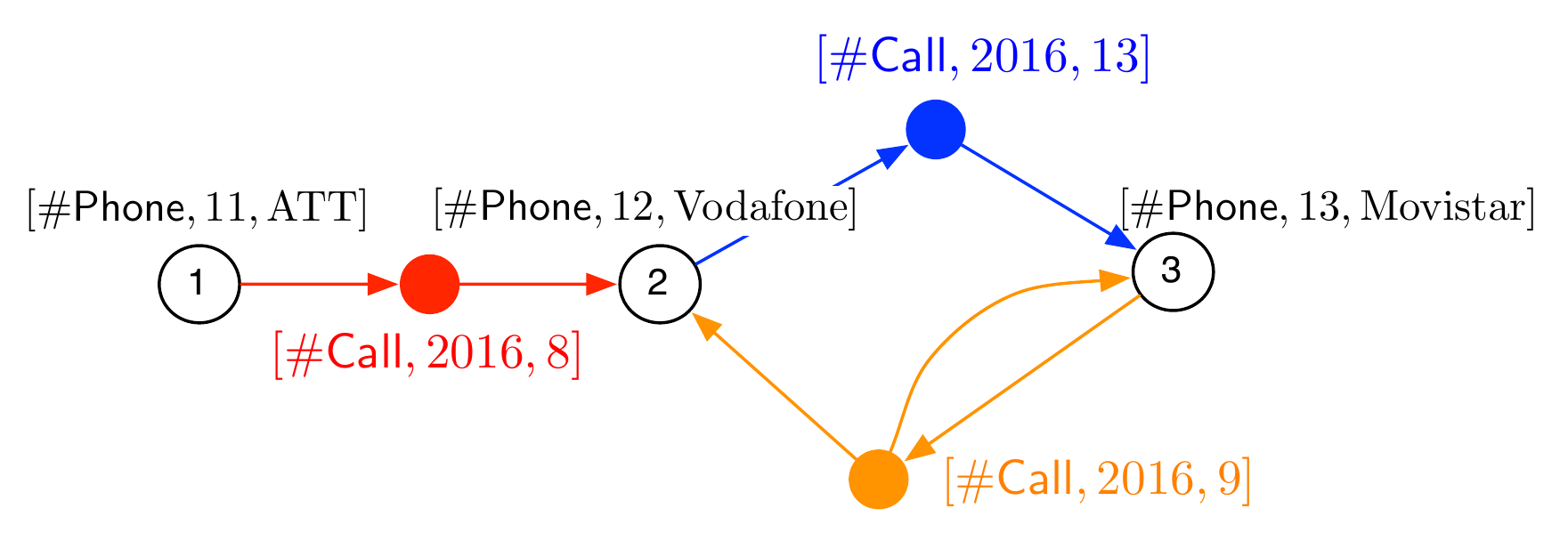}
            \caption{The result of the operation 
$\mbox{\sf Roll-Up}(G, \ab\{ \#\mbox{\sf Phone}\}, \ab \mbox{Time}.(\mbox{Day}\rightarrow\mbox{Year}); \ab \#\mbox{\sf Call},  \mbox{Duration}, \mbox{\sc Sum})$ applied to the  graphoid of Figure~\ref{fig:fig-phone-RU-company}.}
            \label{fig:fig-phone-RU-company+Agg}
 \end{figure}  

  \begin{remark}\rm\label{remark:roll-up}
 To apply the climbing in the roll-up operation 
 to the nodes and edges of all possible types,  the shorthand ``$\ast$'' is used as follows: 
 $\mbox{\sf Roll-Up}(G,\ab \ast, \ab D_c.(\ell_c\rightarrow\ell'_c);\ab \#\mbox{\sf e}, \ab M_1,..., \ab M_k, \ab F_1,..., \ab F_k)$. To  aggregate over all edge types, the notation is 
  $\mbox{\sf Roll-Up}(G,\ast, D_c.(\ell_c\rightarrow\ell'_c); \ast, \ab M_1,..., M_k, \ab F_1,...,  F_k).$
\qed
\end{remark}

 \subsection{Drill-Down}\rm\label{sec:drill-down}
The  $\mbox{\sf Drill-Down}$-operation  does the opposite of    $\mbox{\sf Roll-Up}$,\footnote{Actually, this is true for a sequence of roll-up and drill-down operations such that there are  no slicing or dicing operations (explained in Sections ~\ref{subsec:dice} and ~\ref{subsec:slice})   in-between. However, for the sake of simplicity, and without loss of generality,  in this paper it is assumed that roll-up and drill-down are the inverse of each other.} taking a graphoid to a finer granularity level, along a dimension $D_d$, call it a \textit{descending} dimension, and also operating over a collection of measures, using the same aggregate functions associated with such measures. Note also that, descending from a level $\ell_d$ down to a level  $\ell'_d$ along a dimension $D_d$
is equivalent to climbing from the bottom level of $D_d$, $D_d.\mbox{Bottom}$, to  the level  $\ell'_d$ along   $D_d$. Thus,  the \emph{drill-down of $G$ over the dimensions $M_1,..., M_k$ (using the functions  $F_1,..., F_k$) in hyperedges of type $\#\mbox{\sf e}$, and over 
the  descending dimension $D_d$ from level $\ell_d$  to level $\ell'_d$ in nodes of types $\#\mbox{\sf n}_1,..., \#\mbox{\sf n}_r$ and edges of types $\#\mbox{\sf e}_1,..., \#\mbox{\sf e}_s$},  denoted 
$$\displaylines{\quad \mbox{\sf Drill-Down}(G, \{ \#\mbox{\sf n}_1,..., \#\mbox{\sf n}_r, \#\mbox{\sf e}_1,..., \#\mbox{\sf e}_s\}, \hfill{} \cr \hfill{} D_d.(\ell_d\rightarrow\ell'_d);
\#\mbox{\sf e}, M_1,..., M_k, F_1,..., F_k),\quad}$$  is defined as 

$$\displaylines{\quad\mbox{\sf Aggr}(\mbox{\sf Minimize}(\mbox{\sf Climb}(G, \{\#\mbox{\sf n}_1,..., \#\mbox{\sf n}_r, \#\mbox{\sf e}_1,..., \#\mbox{\sf e}_s\},  \hfill{} \cr \hfill{}  D_d.(\mbox{Bottom}\rightarrow\ell'_d))),\#\mbox{\sf e}, M_1,..., M_k, F_1,..., F_k).\quad}$$

Given the above, in what follows  the discussion  is limited 
to the  $\mbox{\sf Roll-Up}$-operation.

 \subsection{Dice}\label{subsec:dice}   
The $\mbox{\sf Dice}$-operation over a graphoid, produces a subgraphoid that satisfies a 
 Boolean condition $\varphi$ over the available dimension levels. A ``strong'' version is also defined, called the $\mbox{\sf s-Dice}$-operation. 
 In this context, $\varphi$ is a Boolean combination of
 atomic conditions of the form $D.\ell<c$, $D.\ell=c$, and $D.\ell>c$, where $D$ is a dimension, $\ell$ is a level in that dimension, and $c\in  dom(D.\ell)$.
 The expression $\varphi$ can be written in disjunctive normal form as $$\bigvee_k\bigwedge_l \varphi_{kl},$$ where all $\varphi_{kl}$ are atomic conditions.
 
 Before giving the definition of the $\mbox{\sf Dice}$-operation, it must be explained what does it mean that   a hyperedge $e$ in a graphoid \emph{satisfies} $\varphi$, denoted $e\models \varphi$. For this,  interpreting conjunction and disjunction in  
 the usual way, it suffices to define $e\models \varphi_{kl}$ for the atomic formulas that appear in  $\varphi$.   Thus, 
   $\varphi_{kl}$ cannot be evaluated in $e$ if the label of $e$ does not contain 
 information on dimension $D$ at level $\ell$. Otherwise, $\varphi_{kl}$ can be evaluated in $e.$ 
  Let  $\varphi_{kl}$ be $D.\ell<c$, $D.\ell=c$ or $D.\ell>c$;  $\varphi_{kl}$ is \emph{not false} in $e$ 
 if it can be evaluated in $e$ and is true, or if it cannot be evaluated in $e$.  The notion of $\varphi_{kl}$ being \emph{not 
 false} in a node $n$  adjacent to $e$ (that is, $n\in Adj(e)$) is defined analogously.
 Finally,   $e\models \varphi_{kl}$ if $\varphi_{kl}$ is not false in $e$ and not false in all $n\in Adj(e)$.
 
  \begin{definition}[Dice]\rm\label{def:dice}    Assume a $(D_1.\ell_1,..., D_d.\ell_d)$-graphoid $G$ is given as  $G=(N,\tau_N,\lambda_N, E, \tau_E, \lambda_E)$.   
Let $\varphi$ be a Boolean combination of equality and inequality constraints that involve, on the one hand
dimension levels $\ell'_1,..., \ell'_d$ 
(equal or higher than $\ell_1,..., \ell_d$ in the dimension schemas $\sigma(D_1),\ab ..., \ab \sigma(D_d)$, respectively), and on the other hand, constants from $dom(D_1.\ell'_1),..., dom(D_d.\ell'_d)$.  The \emph{dice over $G$ on the condition $\varphi$}, denoted  
$\mbox{\sf Dice}(G, \varphi),$ produces a subgraphoid of $G$, whose nodes are the nodes of $G$ and whose edges satisfy the conditions expressed by $\varphi$. When  an hyperedge  does not satisfy $\varphi$,
the whole hyperedge is deleted from the graph and thus, it does  not belong to  $\mbox{\sf Dice}(G, \varphi).$ All other edges of $G$ belong to $\mbox{\sf Dice}(G, \varphi).$  
If two edges in $G$ have the same set of adjacent nodes   and one of them is deleted from $G$ in $\mbox{\sf Dice}(G, \varphi),$ then both of them are deleted in $G$ to obtain the \emph{strong dice over $G$ on the condition $\varphi$}, denoted $\mbox{\sf s-Dice}(G, \varphi).$
  \qed
\end{definition}

\ignore{
\begin{example}\rm\label{ex:dice}
Applying  the operation 
$\mbox{\sf Dice}(G, \mbox{Phone}.\mbox{Operator}\not=\mbox{ATT})$ to the graphoid depicted in Figure~\ref{fig:fig-phone-RU-company}, results in the graphoid of Figure~\ref{fig:fig-phone-RU-company+Dice}(a). 
In this case,  the  result would be the same as the one obtained after  applying  $\mbox{\sf s-Dice}(G, \mbox{Phone}.\mbox{Operator}\not=\mbox{ATT})$.
 Applying $\mbox{\sf Dice}(G, \mbox{Time}.\mbox{Month}=5/2016)$  over the  graphoid in Figure~\ref{fig:fig-phone-RU-company+Dice}(a),  results in the  graphoid shown in Figure~\ref{fig:fig-phone-RU-company+Dice}(b).
\qed
\end{example}

   \begin{figure}[h]
            \centering
            \includegraphics[scale=0.6]{fig-phone-RU-company+Dice}
            \caption{(a) The result of the operation 
$\mbox{\sf Dice}(G, \mbox{Phone}.\mbox{Operator}\not =\mbox{ATT})$ to the graphoid depicted in Figure~\ref{fig:fig-phone-RU-company}; 
 (b)  The result of applying $\mbox{\sf Dice}(G, \mbox{Time}.\mbox{Month}=5/2016)$ to (a).}
            \label{fig:fig-phone-RU-company+Dice}
 \end{figure}  
}
 
 \subsection{Slice}\label{subsec:slice}    
  Intuitively, the 
$\mbox{\sf Slice}$ operation    eliminates the references to a dimension in a graphoid.  The formal definition follows.
 
\begin{definition}[Slice]\rm\label{def:slice}   Assume a $(D_1.\ell_1,..., D_d.\ell_d)$-graphoid $G$ is given as  $G=(N,\tau_N,\lambda_N, E, \tau_E, \lambda_E)$.   
Let $D_s$ be a dimension that appears in some nodes and/or hyperedges of $G$.  
Let $M_1,..., M_k$ be dimensions that appear in the hyperedges of $G$. These dimensions play the role of measure dimensions. It is assumed that  aggregate functions $F_1,..., F_k$ are associated with them.
The \emph{slice of the dimension $D_s$ from $G$ over the dimensions $M_1,..., M_k$ (using the functions  $F_1,..., F_k$)},  denoted $\mbox{\sf Slice}(G, D_s; M_1,..., M_k, F_1,..., F_k),$  is defined as the roll-up operation up to the level  $D_s.All$ over the dimensions $M_1,..., M_k$ (using the functions  $F_1,..., F_k$). Formally, this slice operation is defined as 
$\mbox{\sf Roll-Up}(G, \ast, D_s.(\ell_s\rightarrow All); \ast,  M_1,..., M_k, F_1,..., F_k). \quad$ \qed
\end{definition}
 
 \ignore{
\begin{example}\rm\label{ex:slice}
Applying the operation $\mbox{\sf Slice}(G,  \mbox{Time}; \mbox{Duration}, \mbox{\sc Sum})$ to the graphoid depicted
 in Figure~\ref{fig:fig-phone-RU-company},
results in  the graphoid of Figure~\ref{fig:fig-phone-RU-company+Slice}.  \qed
\end{example}

   \begin{figure}[h]
            \centering
            \includegraphics[scale=0.7]{fig-phone-RU-company+Slice.pdf}
            \caption{The result of the operation 
$\mbox{\sf Slice}(G,  \mbox{Time}, \mbox{Duration}, \mbox{\sc Sum})$ on the graphoid of Figure~\ref{fig:fig-phone-RU-company}.}

 \label{fig:fig-phone-RU-company+Slice}
 \end{figure}  
}

 \subsection{Node-delete}\label{subsec:node-delete}   
 The $\mbox{\sf n-Delete}$-operation over a graphoid,  deletes 
 all nodes of a certain type and delete, in the source and 
 target set of all edges, the nodes of this type. Again, although 
 this operation is not present in classic OLAP, it is needed 
 to simulate the classic OLAP slice operation, as 
 will become clear in  Section~\ref{sec:olap-graph-equiv}.
 
 \begin{definition}[Node-delete]\rm\label{def:node-delete}  Assume a $(D_1.\ell_1,..., D_d.\ell_d)$-graphoid $G$ is given as  $G=(N,\tau_N,\lambda_N, E, \tau_E, \lambda_E)$.  
 Given a node type $\#\mbox{\sf n}$, the \emph{node-delete  over $G$} operation, denoted  
$\mbox{\sf n-Delete}(G, \#\mbox{\sf n}),$ produces a subgraphoid of $G$, 
whose nodes of type $\#\mbox{\sf n}$ are deleted, and such that all edges $e=S\rightarrow T$ are replaced by edges $S^{\#\mbox{\sf n}}\rightarrow T^{\#\mbox{\sf n}}$, where $S^{\#\mbox{\sf n}}$ and $T^{\#\mbox{\sf n}}$ are $S$ and $T$, respectively, minus the nodes of type
${\#\mbox{\sf n}}$. The edges remain of the same type and they keep the same label. 
  \qed
\end{definition}

\begin{example}\rm\label{ex:node-delete}
When a graphoid contains only nodes of one type, as in Figure~\ref{fig:fig-phone-as-graph}, the result of the deletion of a node  is, obviously, the empty graph. In the graphoid of Figure~\ref{fig:fig-star-flower} (explained later), the result of $\mbox{\sf n-Delete}(G, \#\mbox{\sf Location})$ would be a graph with  nodes 2 and 3, where  a hyperedge containing only these nodes would remain, with label  $[\mbox{\sf \#Sales}, 10]$.  
\qed
\end{example}

\section{Classical OLAP Cubes as a Special Case of OLAP Graphs}\label{sec:classical-olap}

This section explains how the classical cube-based OLAP model can be represented  in the graphoid OLAP model.
It is also shown that the classical OLAP-operations $\mbox{\sf Roll-Up}$, $\mbox{\sf Drill-Down}$, $\mbox{\sf Slice}$ and $\mbox{\sf Dice}$ can be simulated by the graphoid OLAP-operations defined in Section~\ref{sec:olap-operations}.

\subsection{A Discussion on Modelling Cubes as Graphoids}

Figure~\ref{fig:ex-cube} illustrates a  typical example  of an OLAP cube with dimensions $(D_1,\ab D_2,\ab D_3)\ab =\ab (Product, \ Location,$ 
$\  Time).$  The cube represents sales amounts of products at certain stores locations (cities) on certain dates (at the lowest level of granularity). 
There are several ways for representing this cube in the graphoid model.  
Figure~\ref{fig:fig-star-flower} shows two  ways of  modelling  
the fact $(Lego, Antwerp, 1/1/2014; 10)$,  which expresses that the sales of Lego in the 
Antwerp store on January 1st, 2014 amount to 10.

 \begin{figure}[thb]
\centering   
  \includegraphics[scale=1]{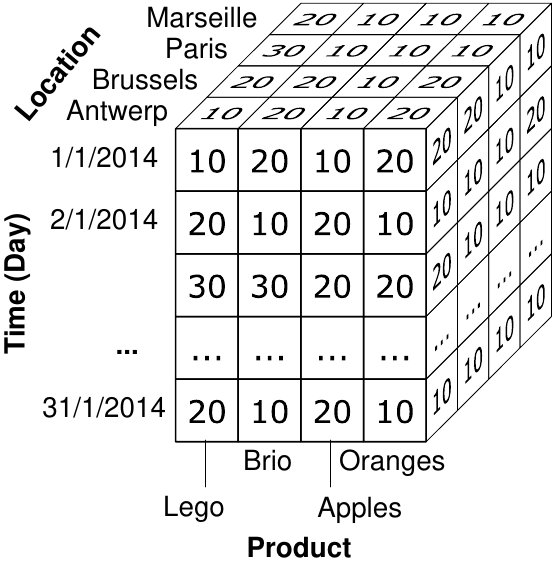}
\caption{An example of a {\sf Sales} data cube with one measure: $\mu_1=sales$.}\label{fig:ex-cube}
\end{figure}

 Figure~\ref{fig:fig-star-flower}(a) shows nodes 1, 2 and 3,  of types 
$\#\mbox{\sf Product}$, $\#\mbox{\sf Location}$ and  $\#\mbox{\sf Time}$, respectively. All  of them have only one attribute, to store the values  $Lego$, $Antwerp$ and $1/1/2014$, call those attributes  $\mbox{ProductVal}$, $\mbox{LocationVal}$ and   $\mbox{TimeVal}$, respectively.  Further, those attributes are dimensions, with an appropriate dimension schema.  The measure information is stored in the  hyperedge $\emptyset\rightarrow \{1,2,3\}$ with label $[\mbox{\sf \#Sales}, 10]$, which has one attribute, namely $\mbox{SalesVal}$, to store the sale amount (10, in this case). Thus, in this approach, each cell of a data cube is modelled by a ``star''-shaped hyperedge.

A more compact representation is shown in Figure~\ref{fig:fig-star-flower}(b).
Here,  there is only one node, of type $\#\mbox{\sf Cube}$ in the graphoid, which represents the data cube. 
This node is labelled $[\#\mbox{\sf Cube},11]$, and has no attribute values (apart from an identifier value).
Cell-coordinates and cell-content are stored in
hyperedges that form loops around  the  node. The fact $(Lego, Antwerp, \mbox{\textit{1/1/2014}}; 10)$ is modelled by a unique hyperedge with label
$[\mbox{\sf \#InCube}, \ab \mbox{Lego}, \ab \mbox{Antwerp}, \ab \mbox{1/1/2014}, \ab 10]$.
Thus, cube facts are represented by a hyperedge of type $\#\mbox{\sf InCube}$ that has four attributes: 
$\mbox{ProductVal}$, $\mbox{LocationVal}$, 
$\mbox{TimeVal}$ and $\mbox{SalesVal}$.

   \begin{figure}[h]
            \centering
            \includegraphics[scale=0.6]{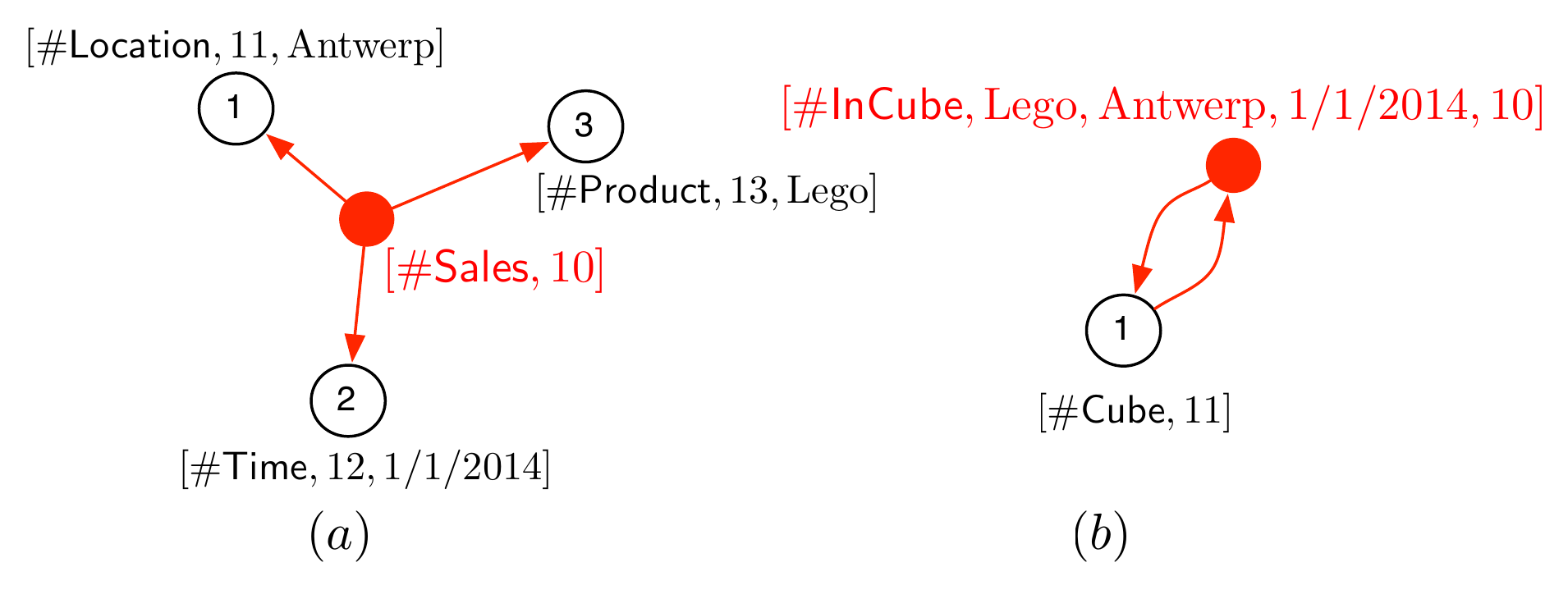}
            \caption{Star-representation of the fact $(Lego, Antwerp, 1/1/2014; 10)$ (a).   Petal-representation of the fact $(Lego, Antwerp, 1/1/2014; 10)$ (b).  }
            \label{fig:fig-star-flower}
 \end{figure}

In between the two  alternatives above, there are, obviously, more modelling possibilities. The next section will show that the graphoid OLAP-operations  presented in Section~\ref{sec:olap-operations}, 
are at least as  powerful as the classical OLAP-operations of the classical   cube model.
The proof will assume  the star-representation of data cubes in the graphoid model (Figure~\ref{fig:fig-star-flower}(a)).

\subsection{Graph- and Classic- OLAP Operations Equivalence}
\label{sec:olap-graph-equiv}
\ignore{ A classical data cube is based on dimensions $D_1,..., D_d$ and on measures $\mu_1,..., \mu_m$.
Each dimension $D_i$ has its domain $dom(D_i)$ and its dimension schema $\sigma(D_i)$ with corresponding dimension instances and roll-up functions (between the levels). Each measure $\mu_j$ has its domain $dom(\mu_j)$ and an associated aggregation function $f_j$.
}
A \emph{(classical) data cube $C$ over   dimensions $D_1,..., D_d$ with  measures $\mu_1,..., \mu_m$}, can then be seen as a   partial  function
$\mu: dom(D_1)\times\cdots\times dom(D_d)\rightarrow dom(\mu_1)\times\cdots\times dom(\mu_m).$ This function maps each ``cell'' of the cube to $m$ values for the measures. A cell of the cube with coordinates $(a_1,..., a_d)\in dom(D_1)\times\cdots\times dom(D_d)$,
that contains values $(c_1,..., c_m)\in dom(\mu_1)\times\cdots\times dom(\mu_m)$, is denoted   by $(a_1,..., a_d;c_1,..., c_m) .$  Below, the ``star-representation'' of a data cube in the graphoid model is formally defined. 

\begin{definition}[Star-graphoid]\rm\label{def:star-of-cube}
Let  $C$ be a data cube over  dimensions $D_1,..., D_d$, with  measures $\mu_1,..., \mu_m$.
The \emph{star-graphoid} of $C$, denoted $\mbox{\sf Star}(C)$, is defined as follows.

\begin{itemize}
\item For $i=1,..., d$,  for each $a_i\in dom(D_i)$, there is a node of type $\#D_i$ with label $[\#D_i, id, a_i]$, where $id$ is a unique node identifier.

\item For each cell $(a_1,..., a_d;c_1,..., c_m) \in dom(D_1)\times\cdots\times dom(D_d)\rightarrow dom(\mu_1)\times\cdots\times dom(\mu_m)$,  there are $m$ hyperedges: for each $j=1,..., m$, there is a hyperedge of type $\#\mu_j$
with an empty source node set and with a target node set
consisting of all nodes labelled  $[\#D_i, id, a_i]$, for $i=1,..., d$, which is labelled $[\#\mu_j, c_j]$.\qed
\end{itemize}
\end{definition} 

Now,  the main theorem of this section is stated.

\begin{theorem}\rm \label{theorem:stronger}
The cube OLAP-operations $\mbox{\sf Roll-Up}$, $\mbox{\sf Drill-Down}$, $\mbox{\sf Slice}$ and $\mbox{\sf Dice}$ can be expressed (or simulated) by OLAP-operations on graphoids. 
\end{theorem}
 
\begin{proof}
 Let $C$ be a data cube, and let $\mbox{\sf Star}(C)$ be its   star-graphoid. The proof  is based on 
 showing that  each of the classical OLAP operations  $\mbox{\sf Roll-Up}$, $\mbox{\sf Drill-Down}$, $\mbox{\sf Slice}$ and $\mbox{\sf Dice}$, over $C$,  can be equivalently applied on $\mbox{\sf Star}(C)$. The semantics for the classical OLAP operations is the one given
 in~\cite{KV17}. \medskip
 
\par\noindent {\bf Roll-Up.} 
For cube data,  a roll-up operation  takes as input a  data cube $C$, a dimension $D_c$ and a level $\ell_i$ in  $\sigma(D_c)$ 
and returns  the aggregation of the original cube along  $D_c$ up to level $\ell_c$ 
for all of the input measures $\mu_1,..., \mu_m$, using aggregate functions $F_1,...,F_m$. 
Assume, without loss of generality,  that  the roll-up starts at the $Bottom$ level, that is, at  $dom(D_c)$.  
Also  assume,  for the sake of clarity of exposition, that $m=1$, that is, that  there is only one measure, call
 it $\mu$, with associated aggregate function $F$.
Now, it will be shown that the roll-up $\mbox{\sf Roll-Up}(C,D_c.\ell_c; \mu, F)$ on the cube $C$
can be simulated on $\mbox{\sf Star}(C)$ by the 
graphoid OLAP-operation $\mbox{\sf Roll-Up}(\mbox{\sf Star}(C), \{ \#\mbox{D}_c   \}, 
D_c.(Bottom\rightarrow\ell_c);\#\mbox{\sf e}_\mu; \mu , F), $ 
where $\#\mbox{D}_c$ is the unique node type in $\mbox{\sf Star}(C)$ that contains information 
on  $D_c$ and where $\#\mbox{\sf e}_\mu$
is the unique edge type that contains measure information on $\mu$.

Let $(a_1,..., a_{c-1}, a_{c+1},..., a_d)$ be an element of 
$dom(D_1)\times\cdots dom(D_{c-1})\times dom(D_{c+1})\times\cdots \times dom(D_{d})$ and suppose that there are 
$r$ values $a_{c,i}$ from $dom(D_c)$ (for $i=1,...,r$) such that $(a_1,..., a_{c-1}, a_{c,i}, a_{c+1},$ 
 $..., a_d; m_i)$ appear in the cube $C$,
and such that all $a_{c,i}$ roll-up to the same element, call it $a_{ru}$, that means $\rho_{D_c.Bottom_k\rightarrow\ell_c}(a)=a_{ru}$. The roll-up on $C$ will replace these $r$ cells by one ``new'' cell which has 
coordinates $(a_1,..., \ab a_{c-1}, \ab a_{ru}, \ab a_{c+1},\ab ..., \ab a_d)$ in $dom(D_1)\times\ab \cdots \ab dom(D_{c-1})\ab \times \ab dom(D_c.\ell_c)\ab \times \ab dom(D_{c+1})\ab \times\cdots \times dom(D_{d})$, and which contains the aggregated measure $F(\{m_1,..., m_r\})$. In $\mbox{\sf Star}(C)$, each one of these ``new'' cells 
will be represented by a hyperedge. To achieve this,  the  following graphoid OLAP-operation is performed: 
$$\mbox{\sf Roll-Up}(\mbox{\sf Star}(C), \#\mbox{D}_c, D_c.(Bottom\rightarrow\ell_c); \#\mbox{\sf e}_\mu, \mu, F).$$

 To see the correctness of this claim,   the substeps in the above graphoid roll-up are analysed. 
 First,   $\mbox{\sf Climb}(\mbox{\sf Star}(C), \#\mbox{D}_c, D_c.(Bottom\rightarrow\ell_c))$ is performed;  a graphoid  called  $G_1$ is   obtained.  Compared against  $\mbox{\sf Star}(C)$, in $G_1$ all nodes and edges remain the same, except for the nodes 
 of type $\#\mbox{D}_c$, which now contain values at level $\ell_c$.
 Next, a minimisation is performed (to obtain a grouping on $D_c$), which may contract some nodes in $G_1$ into ``roll-up'' nodes.    
 Call the resulting graphoid $G_2$. 
 These roll-up nodes of $G_2$ simulate the ``new'' cells in the cube that store the aggregate information. Finally, 
 $\mbox{\sf Aggr}(G_2, \#\mbox{e}_\mu, \mu, F)$ contracts edges that have the same adjacency set and gives them the aggregated value of $\mu$ as attribute value. 

\medskip
\par\noindent {\bf Drill-Down.} As mentioned above, the drill-down to level $\ell$ can be seen as a roll-up from the $Bottom$ level  to level $\ell$.
Therefore, no proof is needed.   

\medskip
\par\noindent {\bf Slice.} 
On data cubes, the $\mbox{\sf Slice}$-operation takes as input a cube $C$, a dimension $D_s$ and returns a cube in which the dimension $D_s$ is dropped, and   all measures are aggregated over the dropped dimension. To   drop the dimension $D_s$,   a roll-up to the level $All$ in this dimension is needed first, such that its domain becomes a singleton. Thus, to simulate this on $\mbox{\sf Star}(C)$ using graphoid OLAP-operations,   a  climb to the level $All$ in the dimension $D_s$ is  
 performed, and therefore the proof of the roll-up case holds, taking into account that  all nodes representing $D_s$ will contain  
  the  value ``$all$''. Thus, the slice of the cube $C$ is simulated by $ \mbox{\sf Slice}({\sf Star}(C), D_s; \mu, F).$   
 There one step missing, however. When slicing a 
 dimension from a cube $C$, this dimension is deleted. In the case of the graphoid ${\sf Star}(C)$, the nodes of type $\# D_s$ are still 
 present in  $ G_1=\mbox{\sf Slice}({\sf Star}(C), D_s; \mu, F).$  So,   $\mbox{\sf n-Delete}(G_1, \#\mbox{D}_s)$ is needed 
 to delete these nodes.

\medskip
\par\noindent {\bf Dice.}
Intuitively, the $\mbox{\sf Dice}(C, \varphi)$ operation, 
where  $\varphi$  is a Boolean condition over level values and measures,
 selects the cells in a cube $C$ that satisfy  $\varphi$.  
  The resulting cube has the same dimensionality
as the original cube. 
It must be shown that   $\mbox{\sf Dice}(C, \varphi)$  can be simulated by $\mbox{\sf s-Dice}(\mbox{\sf Star}(C),\varphi).$ 
As in Section~\ref{subsec:dice}, take    $$\varphi=\bigvee_k\bigwedge_l \varphi_{kl},$$ 
with $\varphi_{kl}$ of the form   $D.\ell<c$, $D.\ell=c$ or $D.\ell>c$, where $D$ is a 
dimension, $\ell$ is a level in that dimension and $c\in dom(D.\ell)$;
or $\mu<c$, $\mu=c$ or $\mu>c$, where $\mu$ is a measure and $c$ belongs to the domain of that measure. 

Let $(a_1,..., a_d;c_1,..., c_m) \in dom(D_1)\times\cdots\times dom(D_d)\rightarrow dom(\mu_1)\times\cdots\times dom(\mu_m)$ be a cell of $C$
that satisfies $\varphi$. Denote this by $(a_1,..., a_d;c_1,..., c_m)\models \varphi$.
The proof here requires showing that the edges $e_j$, labelled $[\#\mu_j, c_j]$ (that are adjacent to the nodes $[\#D_i, id, a_i]$, for $i=1,..., d$), for $j=1,..., m$, also satisfy $\varphi$.
 From $(a_1,..., a_d;c_1,..., c_m)\models \varphi$ it follows that there exists a $k$ such that for all $l$,   
$(a_1,..., a_d;c_1,..., c_m)\models \varphi_{kl}$ holds. 

If $\varphi_{kl}$ is of the form   $D.\ell<c$, $D.\ell=c$ or $D.\ell>c$, 
then $\varphi_{kl}$ is undefined in the edge label and thus, it is  not false in it. Furthermore, because of the particular definition of 
stars in star-graphoids, where all nodes that are adjacent to an edge $e_j$ carry information on unique dimensions, 
$\varphi_{kl}$ is not false in all adjacent nodes that 
do not contain information on $D.\ell$ and it is true in the unique adjacent node that contains information on $D.\ell$. 
Therefore, the edge $e_j$ satisfies $\varphi_{kl}$.

If $\varphi_{kl}$ is of the form $\mu<c$, $\mu=c$ or $\mu>c$, then $\varphi_{kl}$ evaluates to true on one of the edges $e_j$ 
(that contains information on that measure $\mu$)  and is undefined  on the other edges (that contain information on other measures).
On the adjacent nodes to these edges, the condition $\varphi_{kl}$ is not false (since these nodes do not contain information on any measures).  In both cases, all these edges satisfy $\varphi_{kl}$. This means that the strong dice-operation will keep all these edges.

By a similar reasoning, it can be shown that when $(a_1,..., a_d;c_1,..., c_m)\not\models \varphi_{kl}$,  $e_j\not\models\varphi_{kl}$ holds.

This shows that exactly the edges (labelled $[\#\mu_j, c_j]$) corresponding to cells $(a_1,..., a_d;\ab c_1,\ab ...,\ab  c_m)$, where $\varphi$ is not satisfied are deleted from the graphoid 
$\mbox{\sf Star}(C)$ by the strong dice-operation. This completes the proof.
\medskip

\end{proof}

\section{Case Study and Discussion}
\label{sec:casestudy}

The running example followed so far in this paper will also be  used as a case study, in order to evaluate the hypergraph model against the traditional relational OLAP alternative. The example case  has many interesting characteristics, such as: (a) Normally it involves huge volumes of data facts (i.e., calls); (b) The number of dimensions involved in facts is variable, since  calls may differ from one another  in the  number of participants; (c) It allows  performing  not only the typical OLAP operations described in Section~\ref{sec:olap-operations}, over the fact measures, but also to aggregate the graph elements using graph measures like shortest paths, centrality, and so on.  Therefore, the case study is appropriate for illustrating and discussing the graphoid model usefulness in two situations: (a) The classic OLAP scenario, where the relational model is normally used;  and (b) A \textit{Graph OLAP} scenario, where graph metrics are aggregated.  The hypothesis to be tested here is that, although the relational OLAP alternative works better in scenario  (a), when   facts have a fixed dimensionality (e.g., when all calls in the database involve the same number of participants), the graphoid model is competitive when the number of dimensions is variable, and definitely better for scenario (b), where queries compute aggregations over graph metrics. 

The dataset to analyse consists of group calls between  phone lines, where a line cannot call itself, and the analyst also needs to identify the line who started the call.  The schemas of the background dimensions are the ones in Figure~\ref{fig:dimension-schema}, with small changes that will be explained below. Facts are similar to the ones in Figure~\ref{fig:fig-phone-as-graph}. 

Although  performing an exhaustive experimental study is beyond the scope of this paper, and will be part of future work, this section aims at \textit{analysing the plausibility of the graph model} to become a better solution that the relational model for the kinds of problems where factual data are naturally represented as graphs. For this, the graphoid model are compared against the relational alternative containing exactly the same data.  First, two alternative relational OLAP  representations are implemented on a PostgreSQL database,  and three  synthetic datasets of different sizes are produced and loaded into both representations. Then, the same datasets are loaded into a graph database. Neo4j is  used for this purpose, and queries are written in Cypher, Neo4j's high level query language.\footnote{\url{https://neo4j.com/developer/cypher-query-language/}}

\subsection{Relational Representation}
\label{sec:rolapcase}
Since the relational design may impact in query performance, 
two alternative designs for the fact table are  implemented 
in order to provide a fair comparison. In both cases, the  fact table  schema is the following: ${\sf Calls}\mbox{(CallId, CallerId, Participant,StartTime, EndTime, Duration)}.$  

The mea\-ning of the attributes is:
\begin{itemize}
\item CallId: Call identifier;
\item CallerId: The identifier of the line which initiated the call;
\item StartTime, EndTime: Initial and final instants of the call;
\item Duration: Attribute precomputed as (StartTime - EndTime).
\end{itemize}

Although the schemas are the same in both cases,  the instances differ from each other. In one case, a call between phone $Ph_1, Ph_2,$ and $Ph_3$, initiated by $Ph_1$,   contains the tuples $(\mbox{\textit{1}},Ph_1,Ph_2)$ and $(\mbox{\textit{1}},Ph_1,Ph_3).$ In the other case, a tuple $(\mbox{\textit{1}},Ph_1,Ph_1)$ is added to the other two to indicate that 
$Ph_1$ started the call.  This makes a difference for queries where the user is not interested in who did initiate the call.  In what follows, both relational representations are denoted   {\sf Calls} and {\sf Calls-alt}, respectively.

As expressed above, the background dimensions are the same of Figure~\ref{fig:dimension-schema}. There are two slight differences, however, for practical reasons. First,  for the Time  dimension, the bottom level has granularity {\sc Timestamp}, since the StartTime and EndTime attributes in the fact tables have that granularity. That means, a new level is added to the dimension. Second, in the  Phone  dimension the bottom level is the phone identifier, denoted  Id, which rolls up to the line number, denoted Number. This is because  the caller  and the callee are represented as integers, as usual in real world data warehouses.  The Phone dimension is represented in a single table, keeping the constraints indicated by the hierarchies. This representation (i.e., Star) was chosen to provide a fair comparison.  In summary, the dimension  table schema is ${\sf Phone}\mbox{(Id, Number, Customer, City, Country, Operator)}.$   

\subsection{Graphoid-OLAP Representation}

The logical model for the  graphoid representing the calls (i.e., the base graphoid), is similar to the 
one depicted in Figure~\ref{fig:fig-phone-as-graph}.   There are two main entity  nodes, namely  $\#\mbox{{\sf Phone}}$ and 
$\#\mbox{{\sf Call}}$, to represent call facts. These are linked through edges  labelled 
$\#\mbox{{\sf creator}}$ and $\#\mbox{{\sf receiver}},$ the former going from the phone that initiated the call, to the node 
representing such call. Background dimensions are  represented in the same graph, using the entity nodes 
$\#\mbox{{\sf Operator}},$  $\#\mbox{{\sf User}},$ $\#\mbox{{\sf City}}$ and $\#\mbox{{\sf Country}}$ for the dimension levels. 
Finally,  dimension levels are linked using the edges of types $\#\mbox{{\sf provided\_by}},$  $\#\mbox{{\sf has\_phone}},$ $\#\mbox{{\sf belongs\_to}}$ and $\#\mbox{{\sf lives\_in}}.$ It can be observed that nodes are not duplicated. 
 
 \ignore{
 Figure~\ref{fig:screen1} shows a portion of the 
 running example implemented in  Neo4j.

   \begin{figure}[h]
    \centering
            \includegraphics[scale=0.31]{neo4jGraph.png}
            \caption{Portion of the Call-graph.}
            \label{fig:screen1}
 \end{figure} 
 }
\subsection{Datasets}

For the relational representation, synthetic datasets  of two different sizes are generated and loaded into a PostgreSQL  database. Table~\ref{tab:rolapdata} depicts the sizes of the datasets. The first column shows the number of tuples in the {\sf Calls} fact table. The second column shows the number of tuples in the   {\sf Calls-alt} fact table. The third column indicates the number of calls (only one column, since the number of  calls is the same in both versions), and the fourth column tells  the number of tuples in the {\sf Phone} dimension table. 

\begin{table}
\caption{Dataset sizes for the relational representation}
\begin{center}
{
\scriptsize
\begin{tabular}{|c|c|c|c|c|}
\hline 
Dataset &   tuples {\sf Calls}  &  tuples {\sf Calls-alt} & calls &  tuples {\sf Phone} \\ 
\hline 
D1 & 293,817& 420,517& 126,700 & 793\\ 
\hline 
D2 & 528,408 & 756,117& 227,709&  4,689\\ 
\hline 
\end{tabular} }
\end{center}
\label{tab:rolapdata}
\end{table}

For the graph representation,   Table 
~\ref{tab:graphdata} depicts the main numbers of elements in the Neo4j graph. 
\begin{table}
\caption{Dataset sizes for the graph representation}
\begin{center}
{
\scriptsize
\begin{tabular}{|c|c|c|c|c|c|}
\hline 
 Dataset&  Phone nodes &  User nodes &   Call nodes &  creator edges &  receiver  edges \\ 
\hline 
D1 & 793 & 500 & 126,700 & 126,700 & 293,817\\ 
\hline 
D2 &4,689  & 3,000& 227,710 & 227,709  & 528,408 \\ 
\hline 
\end{tabular} }
\end{center}
\label{tab:graphdata}
\end{table}

\subsection{Queries}
This section shows how different kinds of complex analytical queries can be expressed and executed over the three  representations described above. Four  kinds of OLAP queries are discussed: (a) Queries where the aggregations are performed for pairs of objects (e.g., phone lines, persons, etc.); (b) Queries where aggregations are performed in groups of $\mbox{N}$  objects, where $\mbox{N}>2$; (c) For (a) and (b), rollups to different dimension levels are performed.; (d) Graph OLAP-style aggregations performed over graph metrics.  The idea of these experiments is to study if, when the queries can take advantage of the graph structure, graphoid-OLAP  queries are more concisely  expressed, and  more efficiently   executed. The impact of $\mbox{N}$ in the relational and the graph representation is also studied. 
The queries are described next. For the sake of space,   only some of the
SQL and Neo4j queries are shown. 

\begin{query}
Average duration of the calls between  groups of N phone lines.  
\end{query}

This query computes all the  $N$-subsets of  lines that participated in some call. 
That means, if a call involves 3 lines, say $Ph_1, Ph_2$ and $Ph_3$, and   $N=2$,
the groups  will be $(Ph_1, Ph_2),$ $(Ph_1, Ph_3),$ and $(Ph_2, Ph_3).$ Figure~\ref{fig:q11} shows the recursive SQL query for 
the first representation alternative.
\begin{query}
Average duration of the calls between  groups of $\mbox{N}$ users.
\end{query}
\ignore{
Figure~\ref{fig:q3-1} shows the Cypher query for N=3. 
Note that in all cases  the queries actually perform a roll-up   to the level User along the Phone dimension.
The relational queries perform this roll-up through a join between the fact and dimension tables. 
In the case of Neo4j  the roll-up
is performed using by pattern matching. That is, the climbing (in the graphoid OLAP model) 
is done by the {\sf MATCH} clause (the climbing path is explicit in this clause), while
the aggregation is performed in the {\sf RETURN} clause. 

\begin{figure}[ht]
\begin{mdframed}
\begin{small}
 \begin{verbatim}
MATCH (u1:User)<-[:has_phone]-(t1:Phone)<-[:creator|:receiver]-
      (c:Call)-[:creator|:receiver]->(t2:Phone)-[:has_phone]->
      (u2:User), (u3:User)<-[:has_phone]-(t3:Phone)<-
      [:creator|:receiver]-(c:Call)
WHERE u1.id < u2.id  AND u2.id < u3.id
RETURN u1.name, u2.name, u3.name, avg(c.duration) 
\end{verbatim}
\end{small}
\end{mdframed}
\caption{Query 2 -  Cypher for $N=3$.}
\label{fig:q3-1}
\end{figure}

}

\begin{query}
Average duration of the calls between  groups of $\mbox{N}$ operators.
\end{query}
This analyses a roll-up to the level Operator, which has less instance members than the level User addressed in Query 2.

\begin{query}
For each pair of Phones in the Calls graph, compute the shortest path between them. 
\end{query}

This query aims   at analysing the  connections between phone line users, and has many real-world applications (for example, to investigate calls made between two persons who use a third one as an intermediary).
From a technical point of view, this is an aggregation over the whole graph, using as a metric the shortest path between every pair of nodes.  
\ignore{Figure~\ref{fig:q4-1} shows the corresponding Cypher query. }

Finally, the following queries combine the computation of graph metrics together with roll-up and dice operations.

\begin{query}
Compute the shortest path between pairs $(p_1,p_2)$ of phone lines, such that $p_1$ corresponds to operator ``Claro'' and  $p_2$ corresponds to operator ``Movistar''.
\end{query}

\begin{query}
Compute the shortest path between pairs $(p_1,p_2)$ of phone lines, such that $p_1$ corresponds to a user from the city of Buenos Aires  and  $p_2$ corresponds to a user from the city of Salta.
\end{query}

\begin{query}
Compute the shortest path between pairs $(p_1,p_2)$ of phone lines, such that $p_1$ corresponds to a user from the city of Buenos Aires.  
\end{query}

\ignore{
\begin{figure}[t]
\begin{mdframed}
\begin{small}
\begin{verbatim}
MATCH (m:Phone),(n:Phone)
WITH m,n WHERE m<>n
MATCH p= shortestPath((m)-[:receiver|:creator *]-(n))
RETURN p, length(p)
\end{verbatim}
\end{small}
\end{mdframed}
\caption{Query 3 -  Cypher expression. }
\label{fig:q4-1}
\end{figure}
 }
\subsection{Results}

Table \ref{tab:results} shows the results of the experiments. The tests were ran on machine with a  i7-6700 processor and  12 GB of RAM, and  250GB disk (actually, a virtual node in a cluster). The execution times are depicted, and are the averages of five runs of each experiment, expressed in seconds. The winning alternatives are marked in \textbf{boldface}, for clarity.

\begin{table}[h!]
\caption{Experimental results (running times in seconds).}
\begin{scriptsize}
\begin{center}
{\begin{tabular}{|c|c|c|c|c|c|c|c|c|c|}
\hline 
 Dataset & {\sf Calls} & {\sf Calls}&  {\sf Calls} &{\sf Calls-alt}& {\sf Calls-alt}& {\sf Calls-alt}& Neo4j & Neo4j& Neo4j\\ 
 &  $N=2$ & $N=3$  & $N=4$ & $N=2$ &   $N=3$ & $N=4$& $N=2$ &  $N=3$ &  $N=4$\\
\hline 
D1-Q1 & \textbf{4.9} &\textbf{7.6} & \textbf{9.5} & 5.4  & 8.7 & 10.6  & 7.3 & 11.2 & 12.5\\ 
\hline 
D1-Q2 & 4.6 & \textbf{11.7}& \textbf{12.9} & \textbf{4.4}& 12.3 & 14.5 & 7 & \textbf{11.7} & 14.8 \\ 
\hline 
D1-Q3 & 6.6 & \textbf{7.3} & \textbf{11.5} & 12.8 & 12.6 & 14.7 & \textbf{3.7}& 10.8 & 15.5 \\ 
\hline 
D1-Q4 & $\infty$ & N/A  & N/A & $\infty$ & N/A  & N/A  & \textbf{185}   & N/A & N/A  \\ 
\hline
D1-Q5 & $\infty$ & N/A  & N/A & $\infty$ & N/A  & N/A  &  \textbf{21}   & N/A & N/A  \\ 
\hline
D1-Q6 & $\infty$ & N/A  & N/A & $\infty$ & N/A  & N/A & \textbf{6}  & N/A & N/A  \\ 
\hline
D1-Q7 & $\infty$ & N/A  & N/A & $\infty$ & N/A  & N/A  &  \textbf{34}  & N/A & N/A  \\ 
\hline\hline 
D2-Q1  & \textbf{9.3}  & \textbf{14.1}   & 15.1  & 10.4    & 16.2   & 17.7  & 15.6 & 17.5& 21.6\\ 
\hline 
D2-Q2  & \textbf{12.9} &\textbf{19} & 20.7  &  14.5 &  24 & 26.8   & 20.2 & 21.6& 24.8 \\ 
\hline 
D2-Q3 &   12.5  & 19.4  & 22.2  &  14.3   & \textbf{14.6} & 22.8   & \textbf{9.3} &  18.7 & 28.4\\ 
\hline 
D2-Q4 & $\infty$ & N/A  & N/A & $\infty$ & N/A  & N/A  & $\infty$   & N/A & N/A  \\
\hline
D2-Q5 & $\infty$ & N/A  & N/A & $\infty$ & N/A  & N/A  &  \textbf{677}   & N/A & N/A  \\ 
\hline
D2-Q6 & $\infty$ & N/A  & N/A & $\infty$ & N/A  & N/A  & \textbf{123}   & N/A & N/A  \\ 
\hline
D2-Q7 & $\infty$ & N/A  & N/A & $\infty$ & N/A  & N/A  & \textbf{924}   & N/A & N/A  \\ 
\hline
\end{tabular} }
\end{center}
\end{scriptsize}
\label{tab:results}
\end{table}

\subsection{Discussion of Results}

In Table \ref{tab:results} it can be seen that running traditional OLAP queries, like Query 1, Query 2 and Query 3, takes approximately the same time in the relational and graphoid models, with a slight advantage for the former. Further, it can be seen that for Queries 2 and 3, which include a roll-up, results are very similar, and even Neo4j wins here in some cases.  In Query 1, which is an aggregation over the fact graph, the relational alternatives work better.\footnote{It is worth noting that Neo4j (and graph databases in general) is a novel database, whose query optimization strategy is still very basic. On the contrary, relational databases are mature technologies, and query optimization is very efficient indeed. Further, for the experiments presented here, the PostgreSQL databases have been tuned to perform in the best possible way. In this sense, Neo4j's performance for typical OLAP queries is, in some sense, penalized.}  However, for  typical Graph OLAP queries (Queries 4 through 7), which aggregate graph metrics, \textit{the graph model shows a dramatical advantage over the relational alternative.}  For Neo4j,  Query 4 does not finish within a reasonable time for the largest of the two datasets (D2) but performance is acceptable for  D1. On the other hand, the relational alternatives do not 
terminate successfully neither for D1 nor for D2.  It is important to make it clear that with an ad-hoc relational design, specifically for graph representation,  it is possible that the performance of  the relational alternative for shortest path aggregations  could be improved, although it will hardly be close to the graph alternative, given the results presented here. However, the intention of this paper is to present a flexible model that can perform efficiently on a variety of situations. In this sense, the tests presented here  suggest that the graphoid data model can be  competitive with the relational model for classic OLAP queries, but is   much better for typical Graph OLAP ones.   

 \section{Conclusion and Open Problems}
     \label{sec:conclu}
This paper  presented a data model for graph analysis based on node- and edge-labelled directed multi-hypergraphs, called graphoids.  A collection of OLAP operations, analogous to the ones that apply to data cubes, was formally defined over graphoids. It was also formally proved that the classic data cube model is a particular case of the graphoid data model. As far as the authors are aware of, this is the first proposal that formally addresses the problem of defining OLAP operations over hypergraphs. Supported by  this proof, it was shown that the graphoid model can be competitive with the relational implementation of OLAP, but clearly much better when graph operations are used  to aggregate graphs. This feature allows  devising a general  OLAP framework that may cope with the flexible needs of modern data analysis, where data may arrive in different forms.  It is worth to remark, once more, that the  experiments presented do not pretend to be exhaustive, but a good general indication of  the plausibility of the approach, and it is clear that the graph data model provides OLAP with a machinery of more powerful tools than the classic cube data model, which is already good news for the OLAP practitioners.
 
Building on the results in this paper, future work includes looking for further graph metrics that can be applied to the graphoid  model, new case studies, and the study of query optimization strategies. Moreover, the approach can also benefit from tools supporting parallel computation with columnar databases as backends. This can further improve the relational OLAP computation, while keeping the properties of the graphoid model for Graph OLAP queries.   

\section*{Acknowledgments}  
Alejandro Vaisman was  supported by a travel grant from Hasselt University (Korte verblijven--inkomende mobiliteit, BOF16KV09). He was also partially supported by 
PICT-2014 Project 0787 and  
PICT-2017 Project 1054. The authors also  thank T. Colloca, S. Ocamica, J. Perez Bodean, and N. Casta\~no, for their collaboration in the data preparation for the experiments.

\bibliographystyle{plain}

\begin{thebibliography}{10}

\bibitem{Angles2012}
R.~Angles.
\newblock {A Comparison of Current Graph Database Models}.
\newblock In {\em Proceedings of {ICDE} Workshops}, pages 171--177, Arlington,
  VA, USA, 2012.

\bibitem{AnglesABHRV17}
R.~Angles, M.~Arenas, P.~Barcel{\'{o}}, A.~Hogan, J.~L. Reutter, and D.~Vrgoc.
\newblock Foundations of modern query languages for graph databases.
\newblock {\em {ACM} Comput. Surv.}, 50(5):68:1--68:40, 2017.

\bibitem{graphOLAP}
C.~Chen, X.~Yan, F.~Zhu, J.~Han, and P.~Yu.
\newblock {Graph OLAP}: a multi-dimensional framework for graph data analysis.
\newblock {\em Knowl. Inf. Syst.}, 21(1):41--63, 2009.

\bibitem{CohenDDHW09}
J.~Cohen, B.~Dolan, M.~Dunlap, J.M. Hellerstein, and C.~Welton.
\newblock {MAD Skills}: New analysis practices for big data.
\newblock {\em Proceedings of the VLDB Endowment}, 2(2):1481--1492, 2009.

\bibitem{OLAPonBigData}
Alfredo Cuzzocrea, Ladjel Bellatreche, and Il-Yeol Song.
\newblock {Data Warehousing and OLAP over Big Data: Current Challenges and
  Future Research Directions}.
\newblock In {\em Proceedings of DOLAP}, pages 67--70, New York, NY, USA, 2013.
  ACM.

\bibitem{GomezKV17}
Leticia~I. G{\'{o}}mez, Bart Kuijpers, and Alejandro~A. Vaisman.
\newblock Performing {OLAP} over graph data: Query language, implementation,
  and a case study.
\newblock In {\em Proceedings of BIRTE, Munich, Germany, August 28, 2017},
  pages 6:1--6:8, 2017.

\bibitem{GKV19}
Leticia~I. G{\'{o}}mez, Bart Kuijpers, and Alejandro~A. Vaisman.
\newblock Analytical queries on semantic trajectories using graph databases.
\newblock {\em {TGIS} Trans. Geog. Inf. Syst.}, 23(5), 2019.

\bibitem{Hartig14}
O.~Hartig.
\newblock Reconciliation of {RDF*} and property graphs.
\newblock {\em CoRR}, abs/1409.3288, 2014.

\bibitem{Kimball1996}
Ralph Kimball.
\newblock {\em The Data Warehouse Toolkit}.
\newblock J. Wiley and Sons, 1996.

\bibitem{KraiemFKRT15}
M.~B. Kraiem, J.~Feki, K.~Khrouf, F.~Ravat, and O.~Teste.
\newblock Modeling and {OLAP}ing social media: the case of twitter.
\newblock {\em Social Netw. Analys. Mining}, 5(1):47:1--47:15, 2015.

\bibitem{KV17}
Bart Kuijpers and Alejandro~A. Vaisman.
\newblock An algebra for {OLAP}.
\newblock {\em Intelligent Data Analysis}, 21(5), 2017.

\bibitem{RehmanWS13}
N.~U. Rehman, A.~Weiler, and M.~H. Scholl.
\newblock {OLAPing} social media: the case of twitter.
\newblock In {\em Advances in Social Networks Analysis and Mining 2013,
  {ASONAM} '13}, pages 1139--1146, Niagara, ON, Canada, 2013.

\bibitem{Robinson13}
I.~Robinson, J.~Webber, and Emil Eifr\'em.
\newblock {\em Graph Databases}.
\newblock O'Reilly Media, 2013.

\bibitem{DBLP:conf/sigmod/TangHYDZ17}
Bo~Tang, Shi Han, Man~Lung Yiu, Rui Ding, and Dongmei Zhang.
\newblock Extracting top-k insights from multi-dimensional data.
\newblock In {\em Proceedings of {ACM SIGMOD}, Chicago, IL, USA, May 14-19,
  2017}, pages 1509--1524, 2017.

\bibitem{VZ14}
A.~A. Vaisman and E.~Zim\'anyi.
\newblock {\em {Data Warehouse Systems: Design and Implementation}}.
\newblock Springer, 2014.

\bibitem{VBV19}
Alejandro Vaisman, Florencia Besteiro, and Maximiliano Valverde.
\newblock “modelling and querying star and snowflake warehouses using graph
  databases.
\newblock In {\em Proceedings of {ADBIS} Conference 2019, Bled, Slovenia, Sept.
  8-11, 2019}, 2017.

\bibitem{Wang2014}
Z.~Wang, Q.~Fan, H.~Wang, K-L. Tan, D.~Agrawal, and A.~El Abbadi.
\newblock Pagrol: Parallel graph {OLAP} over large-scale attributed graphs.
\newblock In {\em Proceeding of {IEEE ICDE}}, pages 496--507, 2014.

\bibitem{graphCube}
Peixiang Zhao, Xiaolei Li, Dong Xin, and Jiawei Han.
\newblock {Graph Cube}: on warehousing and {OLAP} multidimensional networks.
\newblock In {\em Proceedings of {ACM SIGMOD}}, pages 853--864. ACM, 2011.

\end{thebibliography}

\end{document}